\documentclass[draftclsnofoot,onecolumn]{IEEEtran}
\IEEEoverridecommandlockouts \makeatletter
\def\ps@headings{%
\def\@evenhead{\scriptsize\thepage \hfil \leftmark\mbox{}}%
\def\@oddfoot{}%
\def\@evenfoot{}}
\makeatother
\usepackage{graphicx}
\usepackage{color}
\usepackage{times}
\usepackage{fancyhdr}
\usepackage{amssymb,amsmath}
\usepackage{lastpage}
\usepackage{epsfig}
\usepackage[lined,boxed,commentsnumbered, ruled]{algorithm2e}
\usepackage[compress]{cite}
\usepackage{cases}
\usepackage{caption}
\usepackage{longtable}
\usepackage{multirow}
\usepackage{threeparttable}
\usepackage{rotating}
\usepackage{bm}
\usepackage{float}
\usepackage{array}
\usepackage{subfig}
\usepackage{multirow}
\usepackage{geometry}
\usepackage{soul}
\usepackage{tikz}
\usetikzlibrary{matrix}
\usepackage[utf8]{inputenc}
\usepackage{cite}
\usepackage{comment}
\usepackage{etoolbox}
\usepackage{enumerate}
\usepackage[normalem]{ulem}
\usepackage{stfloats}
\usepackage{pgfplots}
\usepackage{pgfplotstable}
\usepackage{mathrsfs}

\usepgfplotslibrary{fillbetween}
\pgfplotsset{compat=1.18}
\usetikzlibrary{arrows.meta, decorations.pathreplacing, positioning}
\newtheorem{remark}{Remark}

\newtheorem{theorem}{Theorem}
\newtheorem{lemma}{Lemma}
\newtheorem{example}{Example}
\newtheorem{corollary}{Corollary}
\newtheorem{defn}{Definition}

\newcolumntype{C}[1]{>{\centering\let\newline\\\arraybackslash\hspace{0pt}}m{#1}}
\makeatletter
\def\ps@headings{%
\def\@oddhead{\mbox{}\scriptsize\rightmark \hfil \thepage}%
\def\@evenhead{\scriptsize\thepage \hfil \leftmark\mbox{}}%
\def\@oddfoot{}%
\def\@evenfoot{}}
\makeatother

\begin{document}
\title{On the Capacity of DNA Labeling in the Single-Label Setting}

\author{
  \IEEEauthorblockN{
    Zihan~Wu\IEEEauthorrefmark{1},
    Qi~Cao\IEEEauthorrefmark{1},
    Ling~Liu\IEEEauthorrefmark{1}
    and Baoming~Bai\IEEEauthorrefmark{2}
  }\\
  \IEEEauthorblockA{
    \IEEEauthorrefmark{1}
    Guangzhou Institute of Technology, Xidian University, Guangzhou, China
  }
  \\
  \IEEEauthorblockA{
    \IEEEauthorrefmark{2}
    State Key Laboratory of Integrated Service Networks,
    Xidian University, Xi’an 710071, China
  }
  \thanks{This paper was presented in part at ISIT 2026~\cite{Wu2026OnTC}.}
}
\maketitle

\begin{abstract}
DNA labeling has attracted increasing attention in biomedical applications, including molecular imaging, diagnostics, and genomic analysis. In a DNA labeling process, a set of DNA sequence patterns, referred to as labels, is designed according to the requirements of a specific application. For each DNA sequence, the labeling process generates an output sequence that records the positions of the labels. DNA sequences with different labeling outputs can therefore be distinguished through the labeling process. To quantify this capability, the labeling capacity is defined as the exponential growth rate of the maximum number of DNA sequences that can be distinguished through the labeling process as the sequence length tends to infinity~\cite{dnalable1}. To date, the labeling capacities of several cases in the single-label setting have been determined.

In this paper, we formulate the labeling process as a deterministic channel and show that its zero-error capacity is equal to the labeling capacity. For a single label, the corresponding channel can be represented by a star graph. Thus, characterizing the labeling capacity of a single label is equivalent to determining the zero-error capacity of the corresponding star graph. We derive the zero-error capacities of all star graphs, thereby providing a complete characterization of the labeling capacities for all single-label cases. Furthermore, we develop a general method for constructing capacity-achieving codes. These results apply to labeling problems over arbitrary finite alphabets and are not restricted to the DNA alphabet. Finally, for a fixed label length, we exactly characterize the range of achievable labeling capacities and identify all single-label structures that attain the minimum and maximum capacities.
\end{abstract}


\section{Introduction}
DNA labeling has been widely applied in biomedical applications. By detecting specific DNA sequence patterns, DNA labeling provides information about genomic regions of interest and their sequence characteristics. For example, fluorescence in situ hybridization (FISH) uses sequence-specific probes to visualize the locations of genomic regions within cells~\cite{cui2016fluorescence}. In genomic analysis, DNA labeling approaches such as optical mapping record the positions of characteristic sequence patterns along long DNA molecules, providing information for genome assembly and structural analysis~\cite{YUAN20202051}. Moreover, sequence-specific labeling techniques have also been employed in molecular diagnostics to facilitate the detection of disease-related genetic markers~\cite{faltin2013current}. 

This use of sequence-specific patterns for information extraction provides a new perspective on information access in DNA data storage. As DNA storage systems continue to scale, efficiently retrieving specific information from large collections of DNA molecules has become increasingly important. Conventional readout approaches typically rely on sequencing DNA molecules to obtain their complete sequence information~\cite{church2012next}. However, for certain access tasks, identifying specific sequence patterns is sufficient without recovering the entire sequences. To address this challenge, several strategies based on sequence-specific recognition, including PCR-based random access~\cite{organick2018random}, hybridization-based search~\cite{bee2021molecular}, and CRISPR-based approaches~\cite{imburgia2025random}, have been developed for information retrieval in DNA storage systems.

The labeling process can be modeled as follows~\cite{dnalable1}. Let
\(\bm{\alpha}_1,\bm{\alpha}_2,\dots,\bm{\alpha}_M\) be \(M\) labels such that no label is a prefix of any other. For a DNA sequence \(\bm{x}\) with length \(\ell(\bm{x})\geq1\), the corresponding output sequence is of the same length and records the occurrences of these labels. Specifically, the symbol at position \(i\), \(i\in\{0,1,\dots,\ell(\bm{x})-1\}\), is \(m\) if \(\bm{\alpha}_m\) starts at that position for some \(m\in\{1,2,\dots,M\}\), and is zero otherwise. The labeling capacity is defined as the exponential growth rate of the maximum number of distinct output sequences that can be generated as the sequence length tends to infinity.
For example, consider a single label \(\bm{\alpha}=\mathrm{\textcolor{red}{CG}}\) and the DNA sequence
\(\bm{x}=\textnormal{AC\textcolor{red}{CG}G\textcolor{red}{CG}ATC}\). The corresponding output sequence is
\(\textnormal{00\textcolor{red}{1}00\textcolor{red}{1}0000}\).
Since two occurrences of \(\textcolor{red}{\mathrm{CG}}\) cannot start at adjacent positions, the output sequences cannot contain two consecutive \(\textcolor{red}{1}\)'s. Therefore, the set of possible output sequences satisfies the \((d,k)\) run-length-limited (RLL) constraint~\cite{marcus2001introduction} with parameters \((d,k)=(1,\infty)\).
Several cases have been characterized in both the single-label and multiple-label settings, while a general characterization of the labeling capacity remains open~\cite{dnalable1}.

The RLL constraint also appears in the study of zero-error capacity.
This observation further motivates the study of labeling capacity from a zero-error perspective.
In particular, Ahlswede \emph{et al.}~\cite{1998} considered a binary channel with one memory, where \(00\) and \(01\) are the only length-two input blocks with disjoint output sets.
They constructed a capacity-achieving code by restricting the codewords to those satisfying the \((1,\infty)\)-RLL constraint.

The zero-error capacity is defined as the maximum rate at which information can be transmitted with zero error.
This concept was introduced by Shannon in 1956~\cite{1956}, where a discrete memoryless channel was represented by a graph whose vertices are letters of the input alphabet. An edge is placed between two vertices if they are indistinguishable at the channel output. Within this framework, Shannon established a
lower bound on the capacity of a circle graph of length $5$ which was later shown to be tight by Lovász in 1979~\cite{1979}. 
Subsequently, the study of zero-error capacity was extended to channels with memory. The zero-error capacities of binary channels with one memory have been characterized in previous works~\cite{1998,2016,2018}. For binary channels with two memories, several classes of channels have been studied~\cite{zhang2024zero,cao2022zero}; but their zero-error capacities have not been completely characterized.

General characterizations of zero-error capacity remain scarce. For discrete memoryless channels, the Lovász theta function~\cite{1979} provides a powerful upper bound on the zero-error capacity. For channels with memory, \cite{oneedge} developed a coding scheme for channels represented by graphs with one edge, yielding lower bounds on their zero-error capacities.

In this paper, we formulate the labeling process as a deterministic channel and show that its zero-error capacity is equal to the labeling capacity. For a single label, the corresponding channel can be represented by a star graph. Thus, characterizing the labeling capacity of a single label is equivalent to determining the zero-error capacity of the corresponding star graph. We derive the zero-error capacities of all star graphs, thereby providing a complete characterization of the labeling capacities for all single-label cases. Furthermore, we develop a general method for constructing capacity-achieving codes. These results apply to labeling problems over arbitrary finite alphabets and are not restricted to the DNA alphabet. Finally, for a fixed label length, we exactly characterize the range of achievable labeling capacities, as illustrated in Fig.~\ref{fig:capacity}, and identify all single-label structures that attain the minimum and maximum capacities.

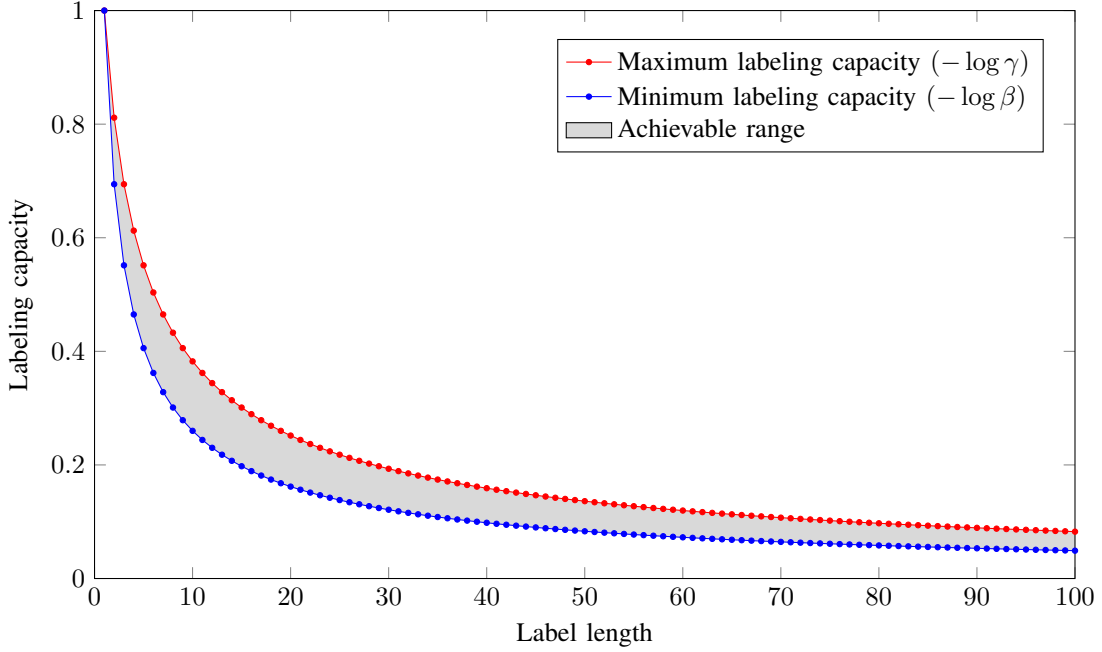
\begin{figure}[!t]
\centering
\begin{tikzpicture}
\begin{axis}[
    width=0.8\linewidth,
    height=0.5\linewidth,
    xlabel={Label length},
    ylabel={Labeling capacity},
    xmin=0,
    xmax=100,
    ymin=0,
    ymax=1,
    xtick={0,10,20,30,40,50,60,70,80,90,100},
    ytick={0,0.2,0.4,0.6,0.8,1.0},
    grid=none,
    legend style={
        at={(0.72,0.95)},
        anchor=north,
        legend columns=1,
        align=left,
        legend cell align=left,
    },
]

\addplot[
    name path=Cmax,
    red,
    thin,
    mark=*,
    mark size=1pt
]
table[x=m,y=Cmax] {data.dat};

\addlegendentry{Maximum labeling capacity $(-\log\gamma)$}

\addplot[
    name path=Cmin,
    blue,
    thin,
    mark=*,
    mark size=1pt
]
table[x=m,y=Cmin] {data.dat};

\addlegendentry{Minimum labeling capacity $(-\log\beta)$}

\addplot[
    draw=none,
    fill opacity=0.15
]
fill between[
    of=Cmax and Cmin
];

\addlegendimage{area legend, fill opacity=0.15}
\addlegendentry{Achievable range}

\end{axis}
\end{tikzpicture}
\caption{The range of labeling capacities in the single-label setting. For a single label \(\bm{\alpha}\) of length \(\ell(\bm{\alpha})\), the labeling capacity \(C_0(\bm{\alpha})\) satisfies \(-\log\beta\leq C_0(\bm{\alpha})\leq-\log\gamma\), where \(\beta\) and \(\gamma\) are the unique positive roots of the equations \(x+x^{\ell(\bm{\alpha})}=1\) and \(x+x^{(\ell(\bm{\alpha})+1)/2}=1\), respectively.}
\label{fig:capacity}
\end{figure}

\section{Definitions and Preliminaries}\label{sec:model}

\begin{table*}[!t]
\centering
\begin{threeparttable}
\caption{Summary of the lemmas and corollaries.}
\label{tab:summary}

\renewcommand{\arraystretch}{1.22}
\setlength{\tabcolsep}{5pt}

\begin{tabular}{|>{\centering\arraybackslash}m{2.3cm}|m{13cm}|}
\hline
\multicolumn{1}{|c|}{Result} &
\multicolumn{1}{c|}{Main conclusion} \\ \hline

Lemma~\ref{co1}
&
Given a label sequence $\bm{a} = \{\bm{\alpha}_m\}_{m=1}^{M}$, two sequences $\bm{x}, \bm{x}' \in \mathcal{X}^{n}$ are distinguishable for the channel $\mathrm{LC}(\bm{a})$ if
$\mathrm{L}_{\bm{a}}(\bm{x}) \neq \mathrm{L}_{\bm{a}}(\bm{x}').$
Equivalently, there exist $m \in \mathbb{Z}[1,M]$ and $i \in \mathbb{Z}[0, n-\ell(\bm{\alpha}_m)]$ such that one of the two substrings
$\bm{x}_{[i;\ell(\bm{\alpha}_m)]}$ and
$\bm{x}'_{[i;\ell(\bm{\alpha}_m)]}$ equals $\bm{\alpha}_m$, while the other does not.
\\ \hline

Lemma~\ref{lem:rate_invariance}

(Lemma~2 in~\cite{2018})
&
We have
$R(\{\mathcal{C}_n\}) = R(\{\mathcal{C}'_t\})$,
i.e.,
$\lim_{n\to\infty} \frac{1}{n}\log |\mathcal{C}_n|
=
\lim_{t\to\infty} \frac{1}{t}\log |\mathcal{C}'_t|$.
\\ \hline

Lemma~\ref{co2}
&
For any label sequence $\bm{a}$,
$C_{\mathrm{L}}(\bm{a})=C_{0}(\bm{a}).$
\\ \hline

Lemma~\ref{lm:lemma2}
&
{
\renewcommand{\arraystretch}{1.1}

Let $\bm{x}$ and $\bm{y}$ be two sequences with $\ell(\bm{x}) \ge 1$.

(1) (Lemma~2 in \cite{oneedge})
If \(\bm{x}\) is both a prefix-unit and a suffix-unit of \(\bm{y}\),
then the sequence \(\bm{x}_{[0;l]}\) is a dividing-unit of \(\bm{x}\) and \(\bm{y}\),
where \(l\) is the largest common divisor of \(\ell(\bm{x})\) and \(\ell(\bm{y})\),
denoted by \(\gcd(\ell(\bm{x}), \ell(\bm{y}))\).

(2) If $\bm{x}$ is a unit of $\bm{y}$, then for any
$i \in \mathbb{Z}[0,\ell(\bm{y})-\ell(\bm{x})]$,
the sequence $\bm{y}_{[i;\ell(\bm{x})]}$ is also a unit of $\bm{y}$.
In particular, $\bm{y}_{[0;\ell(\bm{x})]}$ and
$\bm{y}_{[\ell(\bm{y})-\ell(\bm{x});\,\ell(\bm{x})]}$
are respectively a prefix-unit and a suffix-unit of $\bm{y}$.

(3) For any symbol $c \in \mathcal{X}$ and any sequence $\bm{z} \in \mathcal{X}^n$ with $n\ge1$, let $N_c(\bm{z})$ denote the number of occurrences of $c$ in $\bm{z}$.
Then, for any two units $\bm{u}$ and $\bm{v}$ of $\bm{y}$ with $\ell(\bm{u})=\ell(\bm{v})$, we have $N_c(\bm{u})=N_c(\bm{v})$.
}
\\ \hline

\raisebox{-.9\height}{Lemma~\ref{lm:2}}
&
{
\renewcommand{\arraystretch}{1.1}

Given a sequence $\bm{x}\in\mathcal{X}^n$ with $n \ge 1$ and $l\in\mathbb{Z}[1,\ell(\bm{x})]$, the following statements are equivalent: 

(1) \(\bm{x}_{[0; l]}\) is a prefix-unit of \(\bm{x}\); 

(2) \(\bm{x}_{[\ell(\bm{x})-l; l]}\) is a suffix-unit of \(\bm{x}\); 

(3) $\bm{x}_{[0; \ell(\bm{x}) - l]}=\bm{x}_{[l; \ell(\bm{x}) - l]}$.
}
\\ \hline

Lemma~\ref{lemma:main}
&
{
\renewcommand{\arraystretch}{1.1}
Let $\bm{x} \in \mathcal{X}^n$ with $n \ge 1$.

(1) $D(P(\bm{x})) = P(\bm{x})$; 

(2) For any $i\in\mathbb{Z}[0,n-\ell(P(\bm{x}))]$, we have $\bm{x}_{[i;\ell(P(\bm{x}))]}=P(\bm{x})$ if and only if $i\bmod\ell(P(\bm{x}))=0$.

(3) \(S(\bm{x})\) is the longest sequence $\bm{u}$ such that $\bm{u}\prec_{\mathrm p}\bm{x}$
    and $\bm{u}\prec_{\mathrm s}\bm{x}.$

(4) $|\mathscr{S}(\bm{x})|\leq\ell(\bm{x})$, with equality if and only if $\ell(P(\bm{x}))=1$.
}
\\ \hline

Lemma~\ref{co22}
&
Define
$\mathscr{S}'(\bm{\alpha})\triangleq
\left\{\bm{x}\mid
\bm{x}\prec_{\mathrm{p}}\bm{\alpha},\
\bm{x}\prec_{\mathrm{s}}\bm{\alpha}
\right\}$
and
$\mathscr{S}''(\bm{\alpha})\triangleq
\left\{(P(\bm{\alpha}))^k\circ\bm{\theta}\mid
k\in\mathbb{Z}[0,w-1]\right\}
\cup
\Theta\setminus\{\bm{\theta}\}.$
Then, $\mathscr{S}(\bm{\alpha})=\mathscr{S}'(\bm{\alpha})=\mathscr{S}''(\bm{\alpha})$.
\\ \hline

Corollary~\ref{c0001}
&
There exists a unique index in $\mathbb{Z}[1,|\Theta|]$, denoted by $\tilde{k}$, such that $\bm{\theta}^{\langle \tilde{k}\rangle} = \bm{\theta}$.
\\ \hline

Lemma~\ref{lm:0}
&
{
\renewcommand{\arraystretch}{1.1}
(1) For any $\bm{\varphi}\in\Theta$, we have $\bm{\varphi}\prec_{\mathrm{p}}P(\bm{\alpha})$.

(2) If $\ell(\bm{\theta}^{\langle1\rangle})=\ell(P(\bm{\alpha}))-1$, then $\tilde{k}=1$.
}
\\ \hline

Lemma~\ref{lm:01}
&
{
\renewcommand{\arraystretch}{1.1}
Let $\bm{\phi}\in\Phi$.

(1) $\bm{\alpha}_{[0;\ell(\bm{\alpha})-\ell(P(\bm{\alpha}))+2]}\preceq_{\mathrm{p}}\bm{\phi}$.

(2) 
$P(\bm{\alpha})\prec_{\mathrm{p}}\bm{\phi}$.

(3) $|\Phi|\leq\ell(P(\bm{\alpha}))-1$, where equality holds if and only if one of the following conditions holds:

\begin{tabular}[t]{@{}l@{}}
\hspace{1.5em}(a) $\ell(P(\bm{\alpha}))=1$;\\
\hspace{1.5em}(b) $\ell(P(\bm{\theta}))=1$ and $\ell(\bm{\theta})=\ell(P(\bm{\alpha}))-1$.
\end{tabular}
}
\\ \hline

Lemma~\ref{CnB}
&
$\mathcal{B}$ is suffix-free.
\\ \hline

Corollary~\ref{0o0}
&
Every sequence in $\mathcal{B}^*$
has a unique decomposition into a concatenation of elements of $\mathcal{B}$.
\\ \hline

Lemma~\ref{lm:distinguish}
&
Letting $\bm{b} \in \mathcal{B}^{*}$ with $\ell(\bm{b}) \ge \ell(\bm{\alpha})$, we have  $\bm{\alpha}\preceq_\mathrm{p}\bm{b}$.
\\ \hline

\end{tabular}
\end{threeparttable}
\end{table*}

This section introduces the basic definitions and preliminary concepts used in the study of DNA labeling and zero-error capacity.

Let $\mathcal{X}$ be a finite alphabet with $|\mathcal{X}| \ge 2$. A \emph{label} is any finite sequence $\bm{\alpha}$ over $\mathcal{X}$ with $\ell(\bm{\alpha}) > 0$, where $\ell(\bm{\alpha})$ denotes its length. Throughout this paper, $\bm{\alpha}$ denotes an arbitrary but fixed label unless otherwise specified, and all logarithms are taken to base~2, with the base omitted for brevity.
 For integers \(n_1, n_2 \in \mathbb{Z}\) such that \(n_1 \leq n_2\), let \(\mathbb{Z}[n_1, n_2] \triangleq \{i \in \mathbb{Z} : n_1 \leq i \leq n_2\}\).  
Given a sequence $\bm{x} = (x_0, \cdots, x_{n-1}) \in \mathcal{X}^n$, for any integer
$l \in \mathbb{Z}[0,n]$ and index $i \in \mathbb{Z}[0,n-l]$, the subsequence
$\bm{x}_{[i;l]}$ is defined as
\[
\bm{x}_{[i;l]} \triangleq 
\begin{cases} 
\bm\varepsilon, & \text{if } l = 0,\\
(x_i, x_{i+1}, \cdots, x_{i+l-1}), & \text{if } l=1,2,\cdots,n, \\
\end{cases}
\]
where \(\bm{\varepsilon}\) denotes the empty sequence with length \(0\).
Moreover, for sequences $\bm{u} \in \mathcal{X}^{n_1}$ and
$\bm{v} \in \mathcal{X}^{n_2}$, we define their concatenation as
$\bm{u}\circ\bm{v}
\triangleq
(u_0,\cdots,u_{n_1-1},v_0,\cdots,v_{n_2-1}).$
For any sequence $\bm{x}$ and non-negative integer $t$, let
$\bm{x}^{t}$ denote the concatenation of $\bm{x}$ with itself $t$ times,
where $\bm{x}^{0}$ is defined as the empty sequence $\varepsilon$.
For any two sequences $\bm{u}$ and $\bm{v}$, we write
$\bm{u}\preceq_{\mathrm p}\bm{v}$ and
$\bm{u}\preceq_{\mathrm s}\bm{v}$
if $\bm{u}$ is a prefix and a suffix of $\bm{v}$, respectively.
The symbols $\prec_{\mathrm p}$ and $\prec_{\mathrm s}$ denote proper
prefix and proper suffix relations, respectively.\footnote{
A \emph{proper prefix} (resp., \emph{proper suffix}) of a sequence
$\bm{v}$ is a prefix (resp., suffix) $\bm{u}$ of $\bm{v}$ such that $\bm{u}\neq\bm{v}$.
}
To determine the labeling capacity of a single label, we establish several lemmas and corollaries, which we summarize in Table~\ref{tab:summary} for ease of reference.

\subsection{DNA Labeling}

\begin{defn}\label{def:labeling_sequences}
Let $\bm{\alpha}_1, \bm{\alpha}_2, \dots, \bm{\alpha}_M$ be $M$ labels such that no label is a prefix of any other, and let 
$\bm{a}\triangleq
\{\bm{\alpha}_m \}^{M}_{m=1}$
denote the \emph{label sequence}.
For any sequence $\bm{x} \in \mathcal{X}^n$, the \emph{$\bm{a}$-labeling sequence} of $\bm{x}$ is defined as
$\mathrm{L}_{\bm{a}}(\bm{x}) \triangleq (c_0, c_1, \dots, c_{n-1})$,
where, for each $i \in \mathbb{Z}[0,n-1]$, the symbol $c_i$ is given by
\[
c_i =
\begin{cases}
m, & \text{if } i\in\mathbb{Z}[0,n - \ell(\bm{\alpha}_m)]
       \text{ and }
       \bm{x}_{[i;\,\ell(\bm{\alpha}_m)]} = \bm{\alpha}_m, \\[4pt]
0, & \text{otherwise}.
\end{cases}
\]
Since the label sequence $\bm{a}$ is prefix-free, at most one label can match at any position $i$; therefore, $\mathrm{L}_{\bm{a}}(\bm{x})$ is uniquely defined.
\end{defn}

\begin{example}\label{ep:1}
Consider the label sequence $\bm{a}=\{\textnormal{\textcolor{red}{CG}, \textcolor{blue}{A}} \}$, and a sequence 
$\bm{x} = \textnormal{\textcolor{blue}{A}C\textcolor{red}{CG}G\textcolor{red}{CG}\textcolor{blue}{A}TC}$. 
Its $\bm{a}$-labeling sequence is 
$\mathrm{L}_{\bm{a}}(\bm{x}) = \textnormal{\textcolor{blue}{2}0\textcolor{red}{1}00\textcolor{red}{1}0\textcolor{blue}{2}00}$.
\end{example}

\begin{defn}[Definition~1 in \cite{dnalable1}]\label{rm1}
    Given a label sequence $\bm{a}=
\{\bm{\alpha}_m \}^{M}_{m=1}$, 
define $\mathcal{F}_n({\bm{a}})$ as the set of all 
${\bm{a}}$-labeling sequences of length-$n$, that is,
$\mathcal{F}_n({\bm{a}})=\{ \mathrm{L}_{{\bm{a}}}(\bm{x}) \mid \bm{x} \in \mathcal{X}^n \}.$
The \emph{labeling capacity} of ${\bm{a}}$ is defined as
$$C_\mathrm{L}(\bm{a})\triangleq \limsup_{n \to \infty}\frac{\log |\mathcal{F}_n(\bm{a})|}{n}.$$
\end{defn}

\subsection{Zero-Error Capacity}
The DNA labeling problem can be formulated as a zero-error capacity problem. To establish this connection, given a label sequence $\bm{a} = \{\bm{\alpha}_m\}_{m=1}^{M}$, we model the mapping $\mathrm{L}_{\bm{a}}(\cdot)$ in Definition~\ref{def:labeling_sequences} as a \emph{labeling channel}, denoted by $\mathrm{LC}(\bm{a})$, as illustrated in Fig.~\ref{fig:channel}, which also includes the example in Example~\ref{ep:1}. This channel maps an $\mathcal{X}$-ary input sequence of length-$n$ to an $(M+1)$-ary output sequence of the same length. In the following, we introduce several fundamental concepts from zero-error information theory for $\mathrm{LC}(\bm{a})$, based on which we show that the labeling capacity is equivalent to the zero-error capacity of this channel. 

\begin{defn}
For the channel $\mathrm{LC}(\bm{a})$, let \(p(\bm{y}|\bm{x})\) denote the transition probability of output \(\bm{y}\) given input \(\bm{x}\). Two input sequences \(\bm{u}, \bm{v} \in \mathcal{X}^n\) are \emph{distinguishable} for the channel if
\[
\{\bm{y} : p(\bm{y}|\bm{u}) > 0\} \cap \{\bm{y} : p(\bm{y}|\bm{v}) > 0\} = \emptyset.
\]
\end{defn}

\begin{figure}[t]
  \centering
  \begin{tikzpicture}[>=stealth, thick]

    \node[draw, rounded corners=10pt,
          minimum width=4.8cm,
          minimum height=2cm] (box) {};

    \node[align=center] at (box.center)
      {channel: $\mathrm{LC}(\bm{a})$};

    \node[align=center,
          font=\footnotesize,
          yshift=-6mm] at (box.center)
      {example: $\bm{a}=\{\textnormal{\textcolor{red}{CG}, \textcolor{blue}{A}} \}$};

    \draw[->] 
      ([xshift=-3.5cm]box.west) -- (box.west)
      node[midway, above] {$\bm{x}\in\mathcal{X}^n$}
      node[midway, below]
      {\footnotesize example:
       \textcolor{blue}{A}C\textcolor{red}{CG}G\textcolor{red}{CG}\textcolor{blue}{A}TC};

    \draw[->] 
      (box.east) -- ([xshift=3.5cm]box.east)
      node[midway, above] {$\bm{y}=\mathrm{L}_{\bm{a}}(\bm{x})$}
      node[midway, below]
      {\footnotesize example:
       \textcolor{blue}{2}0\textcolor{red}{1}00\textcolor{red}{1}0\textcolor{blue}{2}00};

  \end{tikzpicture}
  \caption{Labeling channel model with an illustrative example.}
  \label{fig:channel}
\end{figure}
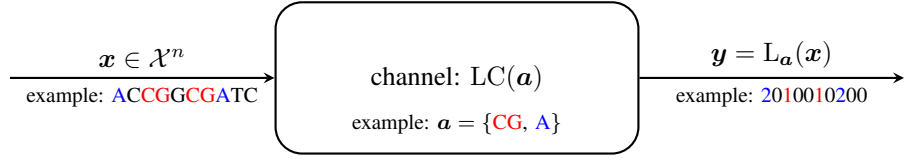

\begin{lemma}\label{co1}
Given a label sequence $\bm{a} = \{\bm{\alpha}_m\}_{m=1}^{M}$, two sequences $\bm{x}, \bm{x}' \in \mathcal{X}^{n}$ are distinguishable for the channel $\mathrm{LC}(\bm{a})$ if
$\mathrm{L}_{\bm{a}}(\bm{x}) \neq \mathrm{L}_{\bm{a}}(\bm{x}').$
Equivalently, there exist $m \in \mathbb{Z}[1,M]$ and $i \in \mathbb{Z}[0, n-\ell(\bm{\alpha}_m)]$ such that one of the two substrings
$\bm{x}_{[i;\ell(\bm{\alpha}_m)]}$ and
$\bm{x}'_{[i;\ell(\bm{\alpha}_m)]}$ equals $\bm{\alpha}_m$, while the other does not.
\end{lemma}

\begin{defn}\label{df:single_graph}
Given a label sequence $\bm{a}=\{\bm{\alpha}_m\}_{m=1}^{M}$, let
$\ell_{\max}=\max_{m\in\mathbb{Z}[1,M]}\ell(\bm{\alpha}_m).$
The labeling channel $\mathrm{LC}(\bm{a})$ is represented as a graph
$G(\bm{a})=(V,E)$, where the vertex set is
$V=\mathcal{X}^{\ell_{\max}}.$
For any two distinct vertices $\bm{u},\bm{v}\in V$, 
$\{\bm{u},\bm{v}\}\in E$ if and only if $\bm{u}$ and $\bm{v}$ are distinguishable
for the channel $\mathrm{LC}(\bm{a})$, i.e.,
$\mathrm{L}_{\bm{a}}(\bm{u})\neq\mathrm{L}_{\bm{a}}(\bm{v}).$
\end{defn}

\begin{remark}
The graphs introduced in this paper and by Shannon \cite{1956} share the same vertex set. 
However, the former connects distinguishable pairs of vertices with edges, whereas the latter connects indistinguishable pairs of vertices. Hence, these two graphs are complements of each other and provide equivalent representations of the labeling channel.
\end{remark}

\begin{defn}
Let $\mathcal{C}_n$ be a set of $n$-length sequences and 
$\{\mathcal{C}_n\}$ be a sequence of such sets 
indexed by $n$. The \emph{asymptotic rate} of $\{\mathcal{C}_n\}$ is 
$R(\{\mathcal{C}_n\}) \triangleq 
\lim\limits_{n \to \infty} \frac{1}{n} \log |\mathcal{C}_n|$,
if it exists, where $|\mathcal{C}_n|$ denotes the cardinality of $\mathcal{C}_n$. If the sequences in $\mathcal{C}_n$ are pairwise 
distinguishable for the graph $G(\bm{a})$, then $\mathcal{C}_n$ is called a \emph{code} of 
length-$n$ for $G(\bm{a})$, and the sequences in $\mathcal{C}_n$ are called \emph{codewords}.
\end{defn}

\begin{defn}
 Let $\{\widehat{\mathcal{C}}_n\}$ be a sequence of codes for the graph $G(\bm{a})$ such that for all $n$, $\widehat{\mathcal{C}}_n$ achieves the largest cardinality of a code of length-$n$ for $G(\bm{a})$. The \emph{zero-error capacity} of the channel is defined as
\[
C_0({\bm{a}}) = \lim_{n \to \infty} \frac{\log |\widehat{\mathcal{C}}_n|}{n}.
\]
\noindent
According to Fekete’s Lemma \cite{fekete1923verteilung}, this limit always exists because $\log |\widehat{\mathcal{C}}_n|$ is \emph{superadditive}, i.e.,
$\log |\widehat{\mathcal{C}}_{m+n}| \geq \log |\widehat{\mathcal{C}}_m| + \log |\widehat{\mathcal{C}}_n|, \forall m, n \geq 1.$
Clearly, $0 \leq C_0({\bm{a}}) \leq \log |\mathcal{X}|$. A sequence of codes $\{\mathcal{C}_n\}$ is said to be \textit{asymptotically optimal} for $G(\bm{a})$ if
$R(\{\mathcal{C}_n\}) = C_0({\bm{a}}).$
\end{defn}

We apply the method in~\cite{2018} to construct a new code based on an exists one. 
Let $\mathcal{C}_n$ be a set of length-$n$ sequences and 
$\{\mathcal{C}_n\}$ be a sequence of such sets indexed by $n$. 
For any $n$, by adding an arbitrary prefix $\bm{u}$
and an arbitrary suffix $\bm{v}$ to all sequences in 
$\mathcal{C}_n$, we obtain a new set of sequences of length 
$t \triangleq n + \ell(\bm{u}) + \ell(\bm{v})$,
denoted by $\mathcal{C}'_t$. 
Let $\{\mathcal{C}'_t\}$ be a sequence of such sets indexed by $t$.

\begin{lemma}[Lemma~2 in~\cite{2018}]
\label{lem:rate_invariance}
We have
$R(\{\mathcal{C}_n\}) = R(\{\mathcal{C}'_t\}),$
i.e.,
$\lim_{n\to\infty} \frac{1}{n}\log |\mathcal{C}_n|
=
\lim_{t\to\infty} \frac{1}{t}\log |\mathcal{C}'_t|.$
\end{lemma}

Next, we show that the DNA labeling problem can be formulated as a zero-error capacity problem.
\begin{lemma}\label{co2}
For any label sequence $\bm{a}$,
$C_{\mathrm{L}}(\bm{a})=C_{0}(\bm{a}).$
\end{lemma}

\begin{IEEEproof}
Let \(\widehat{\mathcal{C}}_n\) be an optimal zero-error code for the graph $G(\bm{a})$ and define 
\(\mathcal{Y}_n \triangleq \{ \mathrm{L}_{\bm{a}}(\bm{x}) \mid \bm{x} \in \widehat{\mathcal{C}}_n \}\). 
Since $\mathrm{L}_{\bm{a}}(\cdot)$ is deterministic and distinct codewords in \(\widehat{\mathcal{C}}_n\) produce distinct outputs, we have \(|\mathcal{Y}_n| = |\widehat{\mathcal{C}}_n|\). 
Moreover, since \(\mathcal{Y}_n \subseteq \mathcal{F}_n(\bm{a})= \{ \mathrm{L}_{\bm{a}}(\bm{x}) \mid \bm{x} \in \mathcal{X}^n\}\), we have 
\(|\widehat{\mathcal{C}}_n| \leq |\mathcal{F}_n(\bm{a})|\).

Conversely, for each output $\bm{y} \in \mathcal{F}_n(\bm{a})$, select one \(\bm{x}(\bm{y}) \in \{ \bm{x} \in \mathcal{X}^n \mid \mathrm{L}_{\bm{a}}(\bm{x}) = \bm{y} \}\), 
and let 
\(\mathcal{X}_n \triangleq \{ \bm{x}(\bm{y}) \mid \bm{y} \in \mathcal{F}_n(\bm{a}) \}\). 
By Lemma~\ref{co1}, \(\mathcal{X}_n\) is a code with \(|\mathcal{X}_n| = |\mathcal{F}_n(\bm{a})|\). 
Since \(\widehat{\mathcal{C}}_n\) is optimal, it follows that 
\(|\mathcal{F}_n(\bm{a})|=|\mathcal{X}_n| \leq |\widehat{\mathcal{C}}_n|\).

Combining the two bounds yields \(|\widehat{\mathcal{C}}_n| = |\mathcal{F}_n(\bm{a})|\), which completes the proof.
\end{IEEEproof}

\subsection{Key Structural Parameters of Labels}

\begin{defn}[Definition~3 in \cite{oneedge}]\label{5}
Let $\bm{x}$ and $\bm{y}$ be two sequences with $\ell(\bm{x}) \ge 1$.
\begin{itemize}
\item 
The sequence $\bm{x}$ is a \emph{unit} of $\bm{y}$ if 
$\ell(\bm{x}) \le \ell(\bm{y})$ and there exists an integer 
$t$ such that 
$y_i = x_{(i - t) \bmod \ell(\bm{x})}$ for any 
$i \in \mathbb{Z}[0, \ell(\bm{y}) - 1]$. 
A \emph{unit} $\bm{x}$ of $\bm{y}$ is called a \emph{prefix-unit} (resp. \emph{suffix-unit}) of $\bm{y}$ if $t = 0$ (resp. $t = \ell(\bm{y}) \bmod \ell(\bm{x})$). 
A prefix-unit $\bm{x}$ is called a \emph{dividing-unit} of $\bm{y}$ if 
$\ell(\bm{x})$ divides $\ell(\bm{y})$.
Clearly, if $\bm{x}$ is a prefix-unit (resp. suffix-unit) of $\bm{y}$,
then $\bm{x}\preceq_{\mathrm p}\bm{y}$ (resp. $\bm{x}\preceq_{\mathrm s}\bm{y}$).

\item 
    Let ${P}(\bm{x})$ denote the shortest prefix-unit of $\bm{x}$, where ${P}(\bm{x})$ is also referred to as the \emph{wagon} of $\bm{x}$. 
   Define ${S}(\bm{x})$ to be the suffix of $\bm{x}$ obtained by removing its prefix ${P}(\bm{x})$, i.e.,
    $\bm{x}={P}(\bm{x})\circ{S}(\bm{x}).$
In particular, if ${P}(\bm{x}) = \bm{x}$, then ${S}(\bm{x}) = \bm{\varepsilon}$, the empty sequence.    
    Clearly,
    $\ell({P}(\bm{x}))\in\mathbb{Z}[1,\ell(\bm{x})], \ell({S}(\bm{x})) \in\mathbb{Z}[0,\ell(\bm{x})-1]$.
    
   \item 
    Let \(D(\bm{x})\) denote the shortest dividing-unit of \(\bm{x}\).  
    Clearly, $\bm{\alpha}=\left(D(\bm{x})\right)^{\ell(\bm{x})/\ell(D(\bm{x}))}.$

    \item Let ${P}(\bm\varepsilon) = {S}(\bm\varepsilon) = D(\bm\varepsilon)= \bm\varepsilon.$
\end{itemize}        
\end{defn}

\begin{example}
\begin{itemize}
    \item Consider the sequence $\bm{x} = 01010101$. 
The sequence $10$ is a unit of $\bm{x}$, but it is neither a prefix-unit nor a suffix-unit. 
The sequence $010101$ is both a prefix-unit and a suffix-unit of $\bm{x}$, but it is not a dividing-unit.  
The sequences $01$ and $0101$ are both dividing-units of $\bm{x}$.  
Moreover, $P(\bm{x}) = D(\bm{x}) = 01$ and $S(\bm{x}) = 010101$.

    \item Consider the sequence $\bm{x} = 01010$. 
The sequences $01$, $0101$, and $01010$ are prefix-units of $\bm{x}$. 
The sequences $10$, $1010$, and $01010$ are suffix-units of $\bm{x}$. 
The sequence $01010$ is a dividing-unit of $\bm{x}$.  
Moreover, $P(\bm{x}) = 01$, $S(\bm{x}) = 010$, and $D(\bm{x}) = 01010$.
\end{itemize}
\end{example}

\begin{lemma}\label{lm:lemma2}
Let $\bm{x}$ and $\bm{y}$ be two sequences with $\ell(\bm{x}) \ge 1$.
\begin{enumerate}[(1)]
\item (Lemma~2 in \cite{oneedge}) If \(\bm{x}\) is both a prefix-unit and a suffix-unit of \(\bm{y}\), 
then the sequence \(\bm{x}_{[0;l]}\) is a dividing-unit of \(\bm{x}\) and \(\bm{y}\), 
where \(l\) is the largest common divisor of \(\ell(\bm{x})\) and \(\ell(\bm{y})\), 
denoted by \(\gcd(\ell(\bm{x}), \ell(\bm{y}))\).

\item
If $\bm{x}$ is a unit of $\bm{y}$, then for any
$i \in \mathbb{Z}[0,\ell(\bm{y})-\ell(\bm{x})]$,
the sequence $\bm{y}_{[i;\ell(\bm{x})]}$ is also a unit of $\bm{y}$.
In particular, $\bm{y}_{[0;\ell(\bm{x})]}$ and
$\bm{y}_{[\ell(\bm{y})-\ell(\bm{x});\,\ell(\bm{x})]}$
are respectively a prefix-unit and a suffix-unit of $\bm{y}$.

\item 
For any symbol $c \in \mathcal{X}$ and any sequence $\bm{z} \in \mathcal{X}^n$ with $n\ge1$, let $N_c(\bm{z})$ denote the number of occurrences of $c$ in $\bm{z}$.
Then, for any two units $\bm{u}$ and $\bm{v}$ of $\bm{y}$ with $\ell(\bm{u})=\ell(\bm{v})$, we have $N_c(\bm{u})=N_c(\bm{v})$.

\end{enumerate}
\end{lemma}

\begin{lemma}\label{lm:2}
Given a sequence $\bm{x}\in\mathcal{X}^n$ with $n \ge 1$ and $l\in\mathbb{Z}[1,\ell(\bm{x})]$, the following statements are equivalent:
\begin{enumerate}[(1)]
    \item \(\bm{x}_{[0; l]}\) is a prefix-unit of \(\bm{x}\);
    \item \(\bm{x}_{[\ell(\bm{x})-l; l]}\) is a suffix-unit of \(\bm{x}\);    
    \item $\bm{x}_{[0; \ell(\bm{x}) - l]} = \bm{x}_{[l; \ell(\bm{x}) - l]}$.
\end{enumerate}
\end{lemma}

\begin{IEEEproof}
Let $\bm{u}=\bm{x}_{[0;l]}$, $\bm{v}=\bm{x}_{[\ell(\bm{x})-l;l]}$,
$p=\lfloor \ell(\bm{x})/l \rfloor$, and
$r=\ell(\bm{x}) \bmod l$.

By Lemma~\ref{lm:lemma2}(2),
Condition~(1) clearly implies Condition~(2).
When Condition~(2) holds, we have 
$\bm{x}=\bm{v}_{[l-r;\,r]}\circ\bm{v}^{p}.$ 
Thus, both $\bm{x}_{[0;\,\ell(\bm{x})-l]}$ and $\bm{x}_{[l;\,\ell(\bm{x})-l]}$ are equal to
$\bm{v}_{[l-r;\,r]}\circ\bm{v}^{p-1}$, i.e., Condition~(3) holds.
Hence, Condition~(2) implies Condition~(3).
When Condition~(3) holds, we have
$$\bm{x}_{[(p-1)l;\,l]}=\bm{x}_{[(p-2)l;\,l]}=\cdots=\bm{x}_{[0;\,l]}=
\bm{u},$$
and $$\bm{x}_{[pl;\,r]}=\bm{x}_{[(p-1)l;\,r]}=\cdots=\bm{x}_{[0;\,r]}=\bm{u}_{[0;\,r]}.$$
Hence, $\bm{x}=\bm{u}^{p}\circ\bm{u}_{[0;\,r]}.$
That is, $\bm{u}$ is a prefix-unit of $\bm{x}$. 
Thus, Condition~(3) implies Condition~(1).
Therefore, Conditions~(1), (2), and (3) are equivalent.
\end{IEEEproof}

\begin{example}
Consider $\bm{x} = 01010$. Then,
$\bm{x}_{[0;2]} = 01$ is a prefix-unit of $\bm{x}$, 
$\bm{x}_{[3;2]} = 10$ is a suffix-unit of $\bm{x}$, 
and $\bm{x}_{[0;3]} = \bm{x}_{[2;3]} = 010$. 
This is consistent with Lemma~\ref{lm:2}.
\end{example}

\begin{figure}[!t]
  \centering
  \begin{tikzpicture}[
    every node/.style={inner sep=1pt},
    scale=0.5
  ]

    \node (AG) at (0,0) {AG};

    \foreach \name/\angle in {
      AA/0, AT/24, AC/48, TA/72, TT/96, TC/120, TG/144,
      CA/168, CT/192, CC/216, CG/240,
      GA/264, GT/288, GC/312, GG/336
    } {
      \node (\name) at (\angle:4) {\name};
      \draw (AG) -- (\name);
    }

  \end{tikzpicture}
  \caption{Star graph of $\mathrm{LC}(\mathrm{AG})$.}
  \label{fig:AG-star}
\end{figure}
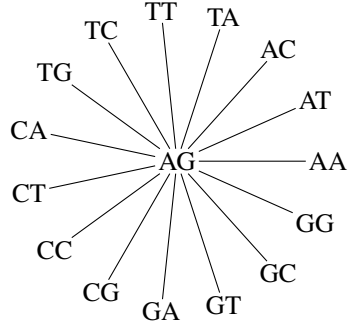

\section{The Labeling Capacity of a Single Label}\label{sec:capacity}

In this section, we determine the labeling capacity of a single label. 
For a single label $\bm{\alpha}$, according to Definition~\ref{df:single_graph}, 
the corresponding graph $G(\bm{\alpha})$ representing the labeling channel 
$\mathrm{LC}(\bm{\alpha})$ is a star graph. 
The unique central vertex is the sequence $\bm{\alpha}$, while all other 
vertices correspond to the sequences in 
$\mathcal{X}^{\ell(\bm{\alpha})}\setminus\{\bm{\alpha}\}$.
For example, when $\mathcal{X}=\{\mathrm{A,C,G,T}\}$, 
the graph $G(\mathrm{AG})$ corresponding to the channel 
$\mathrm{LC}(\mathrm{AG})$ is the star graph shown in Fig.~\ref{fig:AG-star}.
To characterize the labeling capacity of a single label, we introduce several 
structural concepts of $\bm{\alpha}$.

\begin{defn}\label{7}
\begin{itemize}
    \item 
    Let 
    $w\triangleq \left\lfloor \frac{\ell(\bm{\alpha})}{\ell({P}(\bm{\alpha}))} \right\rfloor$ and   
    $\bm{\theta} \triangleq
    \bm{\alpha}_{[w\ell({P}(\bm{\alpha}));\, \ell(\bm{\alpha}) - w\ell({P}(\bm{\alpha}))]},$
    where \(\bm{\theta}\) is referred to as the \emph{tail} of \(\bm{\alpha}\).
    Clearly, the sequence \(\bm{\alpha}\) can be written as a concatenation of $w$ wagons and the tail \(\bm{\theta}\), i.e.,
    $\bm{\alpha}=({P}(\bm{\alpha}))^{w}\circ\bm{\theta}.$
    
\item  
Let $\bm{\theta}^{\langle 0\rangle}\triangleq{P}(\bm{\alpha}) \circ \bm{\theta},$ which is the concatenation of one wagon and the tail.
Clearly, ${P}(\bm{\alpha})$ is a prefix-unit of
$\bm{\theta}^{\langle 0\rangle}$, 
$\bm{\theta}^{\langle 0\rangle}\preceq_\mathrm{p}\bm{\alpha}$, and $\bm{\theta}^{\langle 0\rangle}\preceq_\mathrm{s}\bm{\alpha}$.

\end{itemize}
\end{defn}

\begin{defn}
Let $\bm{x}\in\mathcal{X}^n$ with $n\ge1$.
\begin{itemize}
\item 
Define the $i$-fold application of ${S}(\cdot)$ as
\[
{S}^{i}(\bm{x})
\triangleq
\begin{cases}
\bm{x}, & i = 0,\\
{S}\!\left({S}^{\,i-1}(\bm{x})\right), & i = 1,2,3,\cdots.
\end{cases}
\]

\item
Let
$\mathscr{S}(\bm{x})
\triangleq
\left\{S^k(\bm{x})\,\middle|\,k\in\mathbb{Z}^{+}\right\}.$
Clearly
$\bm{\varepsilon}=S^{|\mathscr{S}(\bm{x})|}(\bm{x})
\prec_{\mathrm{s}}S^{|\mathscr{S}(\bm{x})|-1}(\bm{x})
\prec_{\mathrm{s}}\cdots
\prec_{\mathrm{s}}S(\bm{x})
\prec_{\mathrm{s}}\bm{x}.$

\item
Let
$\Theta\triangleq
\mathscr{S}(\bm{\theta}^{\langle0\rangle})$
and let
$\bm{\theta}^{\langle k\rangle}\triangleq
S^k(\bm{\theta}^{\langle0\rangle}),
k\in\mathbb{Z}[1,|\Theta|].$

\item 
Let
$\Phi\triangleq \left\{ \bm{\alpha}_{[0;\ell(\bm{\alpha})-\ell(\bm{\varphi})]} \mid \bm{\varphi}\in\Theta\setminus\{\bm{\theta}\} \right\}$.
Clearly, for any \(\bm{\phi}\in\Phi\),
$\ell(\bm{\phi})\bmod\ell(P(\bm{\alpha}))\neq0.$
\end{itemize}
\end{defn}

\begin{example}\label{exp}
Consider the following labels:
\begin{itemize}
    \item When $\bm{\alpha} = 00100\,00100$, we have 
${P}(\bm{\alpha})={S}(\bm{\alpha})=D(\bm{\alpha})=00100$,
$\bm{\theta} = \bm{\varepsilon}$, 
$\bm{\theta}^{\langle 0\rangle} = 00100$, and
$\Theta=\{00,0,\bm{\varepsilon}\}$.

    \item When $\bm{\alpha} = 0010\,0010\,0$, we have 
$P(\bm{\alpha}) = 0010$, 
$S(\bm{\alpha}) = 0010\,0$, 
$D(\bm{\alpha}) = 0010\,0010\,0$,
$\bm{\theta} = 0$, 
$\bm{\theta}^{\langle 0\rangle} = 0010\,0$, and
$\Theta=\{00,0,\bm{\varepsilon}\}$.
\end{itemize}
\end{example}

\begin{lemma}[Properties of $P(\cdot)$, $S(\cdot)$, and $D(\cdot)$]
\label{lemma:main}
Let $\bm{x}\in\mathcal{X}^n$ be arbitrary but fixed, where $n\ge1$. Then:
\begin{itemize}
    \item[(1)] $D({P}(\bm{x}))={P}(\bm{x})$.
    \item[(2)] For any $i\in\mathbb{Z}[0,n-\ell(P(\bm{x}))]$, we have $\bm{x}_{[i;\ell(P(\bm{x}))]}=P(\bm{x})$ if and only if $i\bmod\ell(P(\bm{x}))=0$.
    \item[(3)] \(S(\bm{x})\) is the longest sequence $\bm{u}$ such that $\bm{u}\prec_{\mathrm p}\bm{x}$
    and $\bm{u}\prec_{\mathrm s}\bm{x}.$ 
    \item[(4)] $|\mathscr{S}(\bm{x})|\leq\ell(\bm{x})$, with equality if and only if $\ell(P(\bm{x}))=1$.
\end{itemize}
\end{lemma}

\begin{IEEEproof}
\begin{itemize}
\item[(1)] Suppose for contradiction that $D(P(\bm{x})) \neq P(\bm{x})$.
Then $D(P(\bm{x}))$ is a shorter prefix-unit of $\bm{x}$ than $P(\bm{x})$, which contradicts that $P(\bm{x})$ is the shortest prefix-unit of $\bm{x}$ (Definition~\ref{5}).
Thus, $D(P(\bm{x})) = P(\bm{x})$.

\item[(2)]
Clearly, \(i\bmod \ell(P(\bm{x}))=0\) implies that
$\bm{x}_{[i;\ell(P(\bm{x}))]}=\bm{x}_{[0;\ell(P(\bm{x}))]}=P(\bm{x})$.
It remains to show that
if \(i\bmod \ell(P(\bm{x}))\neq0\), then $\bm{x}_{[i;\ell(P(\bm{x}))]}\neq P(\bm{x})$.
Suppose for contradiction that there exists
$i\in\mathbb{Z}[0,n-\ell(P(\bm{x}))]$
such that \(i\bmod \ell(P(\bm{x}))\neq0\)
and $\bm{x}_{[i;\ell(P(\bm{x}))]}=P(\bm{x})$.
Letting $r=i\bmod\ell(P(\bm{x})),$ we have
$\bm{x}_{[r;\ell(P(\bm{x}))]}
=\bm{x}_{[i;\ell(P(\bm{x}))]}
=P(\bm{x}),$
and hence
$$\bm{x}_{[0;\ell(P(\bm{x}))]}=\bm{x}_{[r;\ell(P(\bm{ x}))]}.$$
By Lemma~\ref{lm:2} and Lemma~\ref{lm:lemma2}(2),
we can further obtain that  
$\bm{x}_{[0;r]}$ is both a prefix-unit and a suffix-unit of
$\bm{x}_{[0;\ell(P(\bm{x}))+r]}
=P(\bm{x})\circ\bm{x}_{[0;r]}$. Therefore,
 by Lemma~\ref{lm:lemma2}(1), $\bm{x}_{[0;r]}$ is a dividing-unit of $P(\bm{x}) \circ \bm{x}_{[0;r]}$,
and thus also a dividing-unit of $P(\bm{x})$.
Then, $\ell(D(P(\bm{x}))) \le r < \ell(P(\bm{x}))$,
which contradicts Lemma~\ref{lemma:main}(1).
Thus, if \(i\bmod \ell(P(\bm{x}))\neq0\), then $\bm{x}_{[i;\ell(P(\bm{x}))]}\neq P(\bm{x})$.

\item[(3)] Note that $\bm{x}=P(\bm{x})\circ S(\bm{x})$. By Lemma~\ref{lm:2}, we have $S(\bm{x}) = \bm{x}_{[\ell(\bm{x})-\ell(S(\bm{x}));\,\ell(S(\bm{x}))]}=\bm{x}_{[0;\,\ell(S(\bm{x}))]}$, i.e., $S(\bm{x})\prec_\mathrm{p}\bm{x}$ and $S(\bm{x})\prec_\mathrm{s}\bm{x}$.
Suppose for contradiction that there exists a sequence of length $l > \ell(S(\bm{x}))$ that is also both a proper prefix and a suffix of $\bm{x}$.
Then, by Lemma~\ref{lm:2}, 
$\bm{x}_{[0;\,\ell(\bm{x})-l]}$ is a prefix-unit of $\bm{x}$ with $\ell(\bm{x}) - l < \ell(P(\bm{x}))$, which contradicts that $P(\bm{x})$ is the shortest prefix-unit of $\bm{x}$ (Definition~\ref{5}).
Consequently, \(S(\bm{x})\) is the longest sequence $\bm{u}$ such that $\bm{u}\prec_{\mathrm p}\bm{x}$
    and $\bm{u}\prec_{\mathrm s}\bm{x}.$

\item[(4)] By the definition of $S(\cdot)$, we have
\begin{align*}
    0&=\ell(S^{|\mathscr{S}(\bm{x})|}(\bm{x}))\\
    &=\ell(S^{|\mathscr{S}(\bm{x})|-1}(\bm{x}))
    -\ell(P(S^{|\mathscr{S}(\bm{x})|-1}(\bm{x})))\\
    &=\ell(S^{|\mathscr{S}(\bm{x})|-2}(\bm{x}))
    -\ell(P(S^{|\mathscr{S}(\bm{x})|-2}(\bm{x})))
    -\ell(P(S^{|\mathscr{S}(\bm{x})|-1}(\bm{x})))\\
    &\cdots\\
    &=\ell(\bm{x})
    -\sum_{i=0}^{|\mathscr{S}(\bm{x})|-1}
    \ell(P(S^i(\bm{x})))\\
    &\leq\ell(\bm{x})-|\mathscr{S}(\bm{x})|,
\end{align*}
where equality holds if and only if
$\ell(P(S^i(\bm{x})))=1$ for every
$i\in\mathbb{Z}[0,|\mathscr{S}(\bm{x})|-1]$, i.e., $\ell(P(\bm{x}))=1$.
\end{itemize}
\end{IEEEproof}

\begin{lemma}[Equivalent characterization of $\mathscr{S}(\bm{\alpha})$]\label{co22}
Define
$$\mathscr{S}'(\bm{\alpha})\triangleq
\left\{\bm{x}\mid
\bm{x}\prec_{\mathrm{p}}\bm{\alpha},\
\bm{x}\prec_{\mathrm{s}}\bm{\alpha}
\right\}
\quad\text{and}\quad
\mathscr{S}''(\bm{\alpha})\triangleq
\left\{(P(\bm{\alpha}))^k\circ\bm{\theta}\mid
k\in\mathbb{Z}[0,w-1]\right\}
\cup
\Theta\setminus\{\bm{\theta}\}.$$
Then, $\mathscr{S}(\bm{\alpha})=\mathscr{S}'(\bm{\alpha})=\mathscr{S}''(\bm{\alpha})$.
\end{lemma}
\begin{IEEEproof}
We first show that $\mathscr{S}(\bm{\alpha})=\mathscr{S}'(\bm{\alpha})$. 
By Lemma~\ref{lemma:main}(3), we have
$\bm\varepsilon=S^{|\mathscr{S}(\bm{\alpha})|}(\bm{\alpha})
\prec_{\mathrm p}S^{|\mathscr{S}(\bm{\alpha})|-1}(\bm{\alpha})
\prec_{\mathrm p}\cdots
\prec_{\mathrm p}S(\bm{\alpha})
\prec_{\mathrm p}\bm{\alpha}$
and
$\bm\varepsilon=S^{|\mathscr{S}(\bm{\alpha})|}(\bm{\alpha})
\prec_{\mathrm s}S^{|\mathscr{S}(\bm{\alpha})|-1}(\bm{\alpha})
\prec_{\mathrm s}\cdots
\prec_{\mathrm s}S(\bm{\alpha})\prec_{\mathrm s}\bm{\alpha}$.
Thus, $\mathscr{S}(\bm{\alpha})\subseteq\mathscr{S}'(\bm{\alpha})$.
Suppose for contradiction that $\mathscr{S}(\bm{\alpha})\subset\mathscr{S}'(\bm{\alpha})$, i.e., there exists \(\bm{u}\) such that
$\bm{u}\in\mathscr{S}'(\bm{\alpha})$ and
\(\bm{u}\notin\mathscr{S}(\bm{\alpha})\).
Then, there exists \(k\in\mathbb{Z}[1,|\mathscr{S}(\bm{\alpha})|]\)
such that
\(S^k(\bm{\alpha})\prec_\mathrm{p}\bm{u}
\prec_\mathrm{p}S^{k-1}(\bm{\alpha})\) and
\(S^k(\bm{\alpha})\prec_\mathrm{s}\bm{u}
\prec_\mathrm{s}S^{k-1}(\bm{\alpha})\).
Hence, \(\bm{u}\) is both a proper prefix and suffix of \(S^{k-1}(\bm{\alpha})\), with
$\ell(\bm{u})>\ell(S^k(\bm{\alpha})),$
which contradicts Lemma~\ref{lemma:main}(3).
Therefore, $\mathscr{S}(\bm{\alpha})=\mathscr{S}'(\bm{\alpha})$.

We now show that $\mathscr{S}(\bm{\alpha})=\mathscr{S}''(\bm{\alpha})$.
We begin by showing that $P(S^k(\bm{\alpha})) = P(\bm{\alpha})$ for all $k \in \mathbb{Z}[0,w-2]$.    
This is proved by strong induction on $k$.
    For $k=0$, clearly $P(S^0(\bm{\alpha})) = P(\bm{\alpha})$. 
    Assume that for some $i \in \mathbb{Z}[0,w-3]$, we have $P(S^t(\bm{\alpha})) = P(\bm{\alpha})$ for all $t \in \mathbb{Z}[0,i]$. 
    We now prove that $P(S^{i+1}(\bm{\alpha})) = P(\bm{\alpha})$. Note that $i \le w-3$, and thus
\begin{equation}\label{sj}
\ell(S^{i+1}(\bm{\alpha})) 
= \ell(\bm{\alpha}) - \sum_{t=0}^{i}\ell(P(S^{t}(\bm{\alpha}))) 
= \ell(\bm{\alpha}) - (i+1)\ell(P(\bm{\alpha})) 
\ge 2\ell(P(\bm{\alpha})).
\end{equation}
On the other hand, $P(\bm{\alpha})$ is a prefix-unit of $S^{i}(\bm{\alpha}) = P(\bm{\alpha}) \circ S^{i+1}(\bm{\alpha})$. Therefore, $P(\bm{\alpha})$ is also a prefix-unit of $S^{i+1}(\bm{\alpha})$. Suppose for contradiction that $P(S^{i+1}(\bm{\alpha}))\neq P(\bm{\alpha})$, i.e., there exists a shorter prefix-unit $\bm{v}$ of $S^{i+1}(\bm{\alpha})$ than $P(\bm{\alpha})$. Then, by \eqref{sj}, we have $\ell(S^{i+1}(\bm{\alpha}))\ge\ell(\bm{v})+\ell(P(\bm{\alpha}))$, and thus $S^{i+1}(\bm{\alpha})_{[\ell(\bm{v});\ell(P(\bm{\alpha}))]}=S^{i+1}(\bm{\alpha})_{[0;\ell(P(\bm{\alpha}))]}=P(\bm{\alpha})$. Note that $S^{i+1}(\bm{\alpha})_{[\ell(\bm{v});\ell(P(\bm{\alpha}))]}=\bm{\alpha}_{[(i+1)\ell(P(\bm{\alpha}))+\ell(\bm{v});\ell(P(\bm{\alpha}))]}$. We further have $\bm{\alpha}_{[(i+1)\ell(P(\bm{\alpha}))+\ell(\bm{v});\ell(P(\bm{\alpha}))]}=P(\bm{\alpha})$. By Lemma~\ref{lemma:main}(2), $((i+1)\ell(P(\bm{\alpha}))+\ell(\bm{v}))\bmod\ell(P(\bm{\alpha}))=0,$ and thus $\ell(\bm{v})\bmod\ell(P(\bm{\alpha}))=0,$ which contradicts the assumption that $\ell(\bm{v})<\ell(P(\bm{\alpha}))$. Thus, $P(S^{i+1}(\bm{\alpha})) = P(\bm{\alpha})$. By induction, $P(S^k(\bm{\alpha})) = P(\bm{\alpha})$ for all $k \in \mathbb{Z}[0,w-2]$.
Then we can see that $\{S^k(\bm{\alpha})\mid k\in\mathbb{Z}[1,w-1]\}=\{(P(\bm{\alpha}))^k\circ\bm{\theta}\mid k\in\mathbb{Z}[1,w-1]\}$ and $S^{w-1}(\bm{\alpha})=\bm{\theta}^{\langle0\rangle}$. 
Therefore, 
\begin{align*}
    \mathscr{S}(\bm{\alpha})
&=\{S^k(\bm{\alpha})\mid k\in\mathbb{Z}^+\}\\
&=\{S^k(\bm{\alpha})\mid k\in\mathbb{Z}[1,w-1]\}\cup\{S^{k}(\bm{\alpha})\mid k\ge w, k\in\mathbb{Z}^+\}\\
&=\{(P(\bm{\alpha}))^k\circ\bm{\theta}\mid k\in\mathbb{Z}[1,w-1]\}\cup\{S^{k'}(S^{w-1}(\bm{\alpha}))\mid k'\in\mathbb{Z}^+\}\\
&=\{(P(\bm{\alpha}))^k\circ\bm{\theta}\mid k\in\mathbb{Z}[1,w-1]\}\cup\Theta\\
&=\{(P(\bm{\alpha}))^k\circ\bm{\theta}\mid k\in\mathbb{Z}[0,w-1]\}\cup(\Theta\setminus\{\bm{\theta}\})\\
&=\mathscr{S}''(\bm{\alpha}).
\end{align*}
\end{IEEEproof}

\begin{remark}
By Lemma~\ref{co22}, each application of ${S}(\cdot)$ removes one wagon 
${P}(\bm{\alpha})$ from the start of the label $\bm{\alpha}$.  
The process continues until only $\bm{\theta}^{\langle 0 \rangle}$ remains. However, it does not necessarily hold that
${P}(\bm{\theta}^{\langle 0 \rangle}) = {P}(\bm{\alpha})$. For example, when $\bm{\alpha}=0010\,0010\,0$,
we have $P(\bm{\alpha})=0010$ and $\bm{\theta}^{\langle 0 \rangle}=0010\,0$. Then, ${P}(\bm{\theta}^{\langle 0 \rangle})=001 \neq {P}(\bm{\alpha})$.
\end{remark}

\begin{corollary}\label{c0001}
    There exists a unique index in $\mathbb{Z}[1,|\Theta|]$, denoted by $\tilde{k}$, such that $\bm{\theta}^{\langle \tilde{k}\rangle} = \bm{\theta}$.
\end{corollary}

\begin{lemma}[Properties of $\Theta$]\label{lm:0}
\begin{itemize}
\item[(1)] For any $\bm{\varphi}\in\Theta$, we have $\bm{\varphi}\prec_{\mathrm{p}}P(\bm{\alpha})$.
\item[(2)] If $\ell(\bm{\theta}^{\langle1\rangle})=\ell(P(\bm{\alpha}))-1$, then $\tilde{k}=1$.
\end{itemize}
\end{lemma}

\begin{IEEEproof}
\begin{itemize}
\item[(1)]
By Lemma~\ref{co22}, we have
$\bm{\varphi}\prec_\mathrm{p}\bm{\alpha}$, and thus it suffices to show that
$\ell(P(\bm{\alpha}))>\max_{\bm{\varphi}\in\Theta}\ell(\bm{\varphi})=\ell(\bm{\theta}^{\langle1\rangle})$.
Suppose for contradiction that $\ell(P(\bm{\alpha}))\leq\ell(\bm{\theta}^{\langle1\rangle}).$
Then, 
\begin{equation}\label{am}
    \ell(\bm{\theta}^{\langle 0 \rangle})=\ell(P(\bm{\theta}^{\langle 0 \rangle}))+\ell(\bm{\theta}^{\langle 1 \rangle})\geq\ell(P(\bm{\theta}^{\langle 0 \rangle}))+\ell(P(\bm{\alpha})).
\end{equation}
Thus, $\bm{\theta}^{\langle 0 \rangle}_{[\ell(P(\bm{\theta}^{\langle 0 \rangle}));\ell(P(\bm{\alpha}))]}=\bm{\theta}^{\langle 0 \rangle}_{[0;\ell(P(\bm{\alpha}))]}=P(\bm{\alpha})$.
Since $\bm{\theta}^{\langle 0 \rangle}\preceq\bm{\alpha}$, we further have $\bm{\alpha}_{[\ell(P(\bm{\theta}^{\langle 0 \rangle}));\ell(P(\bm{\alpha}))]}=\bm{\theta}^{\langle 0 \rangle}_{[\ell(P(\bm{\theta}^{\langle 0 \rangle}));\ell(P(\bm{\alpha}))]}=P(\bm{\alpha})$.
By Lemma~\ref{lemma:main}(2), we have $\ell(P(\bm{\theta}^{\langle 0 \rangle}))\bmod\ell(P(\bm{\alpha}))=0,$ 
and thus $\ell(P(\bm{\theta}^{\langle 0 \rangle}))\geq\ell(P(\bm{\alpha}))$. 
However, by \eqref{am},
$$\ell(P(\bm{\theta}^{\langle 0 \rangle}))\leq\ell(\bm{\theta}^{\langle 0 \rangle})-\ell(P(\bm{\alpha}))=\ell(\bm{\theta})<\ell(P(\bm{\alpha})),$$
which is a contradiction. 
Therefore, $\ell(P(\bm{\alpha}))>\ell(\bm{\theta}^{\langle1\rangle}).$

\item[(2)]
If $\ell(\bm{\theta}^{\langle 1\rangle})=\ell(P(\bm{\alpha}))-1$, then
\begin{equation}\label{1-11}
    \ell(P(\bm{\theta}^{\langle 0\rangle}))
=\ell(\bm{\theta}^{\langle 0\rangle})-\ell(\bm{\theta}^{\langle 1\rangle})
=(\ell(P(\bm{\alpha}))+\ell(\bm{\theta}))-(\ell(P(\bm{\alpha}))-1)
=\ell(\bm{\theta})+1.
\end{equation}
Hence,
$P(\bm{\theta}^{\langle 0\rangle})=\bm{\theta}^{\langle 0\rangle}_{[
0;\ell(\bm{\theta})+1]}
=\bm{\theta}\circ\theta^{\langle 0\rangle}_{\ell(\bm{\theta})}$.
Note that \(P(\bm{\alpha})\) is a prefix-unit of
\(\bm{\theta}^{\langle 0\rangle}\). We see that
\(P(\bm{\theta}^{\langle 0\rangle})\) is also a prefix-unit of
\(P(\bm{\alpha})\). Moreover, by Lemma~\ref{lm:lemma2}(2),
$P(\bm{\alpha})_{[\ell(P(\bm{\alpha}))-(\ell(\bm{\theta})+1);\ell(\bm{\theta})+1]}$ is a suffix-unit of
$P(\bm{\alpha})$ and
\begin{align*}
P(\bm{\alpha})_{[
\ell(P(\bm{\alpha}))-(\ell(\bm{\theta})+1);\ell(\bm{\theta})+1]}
&=\bm{\theta}^{\langle 0\rangle}_{[\ell(P(\bm{\alpha}))-(\ell(\bm{\theta})+1);\ell(\bm{\theta})+1]}\\
&=\bm{\theta}^{\langle 0\rangle}_{[\ell(P(\bm{\alpha}))-(\ell(\bm{\theta})+1);\ell(\bm{\theta})]}\circ\theta^{\langle 0\rangle}_{\ell(P(\bm{\alpha}))-1}\\
&\overset{(a)}{=}\bm{\theta}^{\langle 0\rangle}_{[\ell(P(\bm{\alpha}));\ell(\bm{\theta})]}\circ\theta^{\langle 0\rangle}_{\ell(P(\bm{\alpha}))-1}\\
&=\bm{\theta}\circ\theta^{\langle 0\rangle}_{\ell(P(\bm{\alpha}))-1},
\end{align*}
where $(a)$ follows from $\bm{\theta}\circ\theta^{\langle 0\rangle}_{\ell(\bm{\theta})}$ being a prefix-unit of $\theta^{\langle 0\rangle}$.
By Lemma~\ref{lm:lemma2}(3),
$\theta^{\langle 0\rangle}_{\ell(\bm{\theta})}=\theta^{\langle 0\rangle}_{\ell(P(\bm{\alpha}))-1}.$
Hence, $P(\bm{\theta}^{\langle 0\rangle})=P(\bm{\alpha})_{[
\ell(P(\bm{\alpha}))-(\ell(\bm{\theta})+1);\ell(\bm{\theta})+1]}$, and thus
$P(\bm{\theta}^{\langle 0\rangle})$
is both a prefix-unit and a suffix-unit of $P(\bm{\alpha})$.
Then, by Lemma~\ref{lm:lemma2}(1),
$P(\bm{\theta}^{\langle 0\rangle})$
is a dividing-unit of $P(\bm{\alpha})$.
By Lemma~\ref{lemma:main}(1), $D(P(\bm{\alpha}))=P(\bm{\alpha})$, and thus $P(\bm{\theta}^{\langle 0\rangle})=P(\bm{\alpha})$.
Then, 
$\bm{\theta}^{\langle 1 \rangle}=S(\bm{\theta}^{\langle 0 \rangle})=\bm{\theta}=\bm{\theta}^{\langle \tilde{k} \rangle}$, i.e., $\tilde{k}=1$. 
\end{itemize}
\end{IEEEproof}

\begin{lemma}[Properties of $\Phi$]\label{lm:01}
Let $\bm{\phi}\in\Phi$.
\begin{itemize}
\item[(1)] $\bm{\alpha}_{[0;\ell(\bm{\alpha})-\ell(P(\bm{\alpha}))+2]}\preceq_{\mathrm{p}}\bm{\phi}$.
    \item[(2)] $P(\bm{\alpha})\prec_{\mathrm{p}}\bm{\phi}$.
    \item[(3)] $|\Phi|\leq\ell(P(\bm{\alpha}))-1$, with equality if and only if one of the following holds:
\begin{itemize}
    \item[(a)] $\ell(P(\bm{\alpha}))=1$;
    \item[(b)] $\ell(P(\bm{\theta}))=1$ and $\ell(\bm{\theta})=\ell(P(\bm{\alpha}))-1$.
\end{itemize}
\end{itemize}
\end{lemma}

\begin{IEEEproof}
    \begin{itemize}
\item[(1)]
Note that $\bm{\phi}\preceq_\mathrm{p}\bm{\alpha}$. It suffices to show that
$\ell(\bm{\phi})\geq\ell(\bm{\alpha})-\ell(P(\bm{\alpha}))+2$.
By Lemma~\ref{lm:0}(1), we have
$\ell(\bm{\theta}^{\langle1\rangle})\leq\ell(P(\bm{\alpha}))-1$.
If $\tilde{k}=1$, then
\[
\ell(\bm{\phi})\ge\ell(\bm{\alpha})-\ell(\bm{\theta}^{\langle2\rangle})>\ell(\bm{\alpha})-\ell(\bm{\theta}^{\langle1\rangle})
\geq\ell(\bm{\alpha})-\ell(P(\bm{\alpha}))+1.
\]
If $\tilde{k}>1$, by Lemma~\ref{lm:0}(2), we have
\[
\ell(\bm{\phi})\geq\ell(\bm{\alpha})-\ell(\bm{\theta}^{\langle1\rangle})
>\ell(\bm{\alpha})-\ell(P(\bm{\alpha}))+1.
\]
Therefore,
$\ell(\bm{\phi})>\ell(\bm{\alpha})-\ell(P(\bm{\alpha}))+1$, i.e., $\ell(\bm{\phi})\ge\ell(\bm{\alpha})-\ell(P(\bm{\alpha}))+2$.
   
        \item[(2)]
We first show that $\tilde{k}\in\mathbb{Z}[2,|\Theta|]$ implies $w\ge2$.
When $w=1$, we have $\bm{\theta}^{\langle 0 \rangle}=\bm{\alpha}$, and thus
$P(\bm{\theta}^{\langle 0 \rangle})=P(\bm{\alpha}).$
Then, 
$\bm{\theta}^{\langle 1 \rangle}=S(\bm{\theta}^{\langle 0 \rangle})=\bm{\theta}=\bm{\theta}^{\langle \tilde{k} \rangle}$, i.e., $\tilde{k}=1$.
Therefore, $w=1$ implies $\tilde{k}=1$, or equivalently,
$\tilde{k}\in\mathbb{Z}[2,|\Theta|]$ implies $w\ge2$.

We now show that $P(\bm{\alpha})\prec_{\mathrm{p}}\bm{\phi}$.
Since $P(\bm{\alpha})\preceq\bm{\alpha}$ and $\bm{\phi}\preceq\bm{\alpha}$, we only need to show
$\ell(\bm{\phi})>\ell(P(\bm{\alpha}))$.
If $\tilde{k}=1$, then $$\ell(\bm{\phi})>\ell(\bm{\alpha})-\ell(\bm{\theta}^{\langle 1\rangle})= \ell(\bm{\alpha})-\ell(\bm{\theta})\geq\ell(P(\bm{\alpha})).$$
If $\tilde{k}>1$, then $w\geq2$, and thus
$$\ell(\bm{\phi})\geq\ell(\bm{\alpha})-\ell(\bm{\theta}^{\langle 1\rangle})>
\ell(\bm{\alpha})-\ell(\bm{\theta}^{\langle 0\rangle})=(w-1)\ell(P(\bm{\alpha}))\geq \ell(P(\bm{\alpha})).$$
Therefore, $\ell(\bm{\phi})>\ell(P(\bm{\alpha}))$.

\item[(3)]
We first show that $|\Phi|\leq\ell(P(\bm{\alpha}))-1$, with equality only if Condition (a) or (b) holds.
If $\bm{\theta}^{\langle1\rangle}=\bm\varepsilon$, then $\Phi=\varnothing$, and hence $|\Phi|\leq\ell(P(\bm{\alpha}))-1$, where equality implies Condition (a).
If $\bm{\theta}^{\langle1\rangle}\neq\bm\varepsilon$, then
\begin{equation}\label{phi}
    |\Phi|
    =|\Theta|-1
    =|\mathscr{S}(\bm{\theta}^{\langle1\rangle})|
    \overset{(a)}{\leq}\ell(\bm{\theta}^{\langle1\rangle})
    \overset{(b)}{\leq}\ell(P(\bm{\alpha}))-1,
\end{equation}
where $(a)$ follows from Lemma~\ref{lemma:main}(4), with equality only if $\ell(P(\bm{\theta}^{\langle1\rangle}))=1$, and $(b)$ follows from Lemma~\ref{lm:0}(1). Moreover, by Lemma~\ref{lm:0}(2), equality in $(b)$ implies $\bm{\theta}^{\langle1\rangle}=\bm{\theta}$. Thus, both equalities in \eqref{phi} can hold only if Condition (b) holds.

It remains to show that Conditions (a) and (b) are sufficient for $|\Phi|=\ell(P(\bm{\alpha}))-1$.
If Condition (a) holds, then $\bm{\theta}=\bm{\theta}^{\langle1\rangle}=\bm\varepsilon$, and thus $|\Phi|=0=\ell(P(\bm{\alpha}))-1$.
If Condition (b) holds, Lemma~\ref{lm:0}(2) gives $\bm{\theta}=\bm{\theta}^{\langle1\rangle}$. Since $\ell(P(\bm{\theta}))=1$, Lemma~\ref{lemma:main}(4) further yields
\begin{equation*}
|\Phi|
=|\Theta|-1
=|\mathscr{S}(\bm{\theta}^{\langle1\rangle})|
=|\mathscr{S}(\bm{\theta})|
=\ell(\bm{\theta})
=\ell(P(\bm{\alpha}))-1.
\end{equation*}
\end{itemize}
\end{IEEEproof}

Now, we determine the zero-error capacity of star graph $G(\bm{\alpha})$.

\begin{defn}\label{cn}
Let
$\mathcal{B}\triangleq
\{P(\bm{\alpha})\}\cup\Phi\cup\mathcal{S},$
where
$\mathcal{S}\triangleq
\left\{
\bm{\alpha}\circ s^t
\mid t\in\mathbb{Z}^{+}
\right\}$
with an arbitrary but fixed
$s\in\mathcal{X}\setminus\{\alpha_{\ell(P(\bm{\alpha}))-1}\}$.
Clearly, for any two distinct elements
\(\bm b,\bm b'\in\mathcal{B}\), either
$\bm b\prec_{\mathrm p}\bm b'$
or
$\bm b'\prec_{\mathrm p}\bm b.$

Let $\mathcal{B}^{*}$ denote the Kleene closure of $\mathcal{B}$, i.e.,
$\mathcal{B}^{*}\triangleq
\bigcup_{m=0}^{\infty}\mathcal{B}^{m},$
where $\mathcal{B}^{m}$ denotes the set of all sequences obtained by
concatenating $m$ elements from $\mathcal{B}$, with
$\mathcal{B}^{0}=\emptyset$.
\end{defn}

\begin{theorem}\label{thm}
The zero-error capacity of \(G(\bm\alpha)\) is given by
$C_\mathrm{0}({\bm{\alpha}})= -\log \mu,$
where \(\mu\) is the only positive root of the
equation
\[
x^{\ell({P}(\bm{\alpha}))}
+\sum_{\bm\phi\in\Phi}x^{\ell(\bm\phi)}
+\frac{x^{\ell(\bm{\alpha})+1}}{1-x}=1.
\]
\end{theorem}

\begin{lemma}\label{CnB}
$\mathcal{B}$ is suffix-free.
\end{lemma}

\begin{IEEEproof}
Let $\bm{b}'\in\mathcal{B}$ be arbitrary.
Suppose for contradiction that there exists
$\bm{b}\in\mathcal{B}$ such that
$\bm{b}\prec_{\mathrm{s}}\bm{b}'.$
Since $\ell(\bm{b})<\ell(\bm{b}')$, by Definition~\ref{cn},
we have $\bm{b}\prec_{\mathrm{p}}\bm{b}'.$
Then, there exist strings \(\bm{u}\) and \(\bm{u}'\) such that
$\bm{b}'=\bm{u}\circ\bm{b}=\bm{b}\circ\bm{u}'.$
By Lemma~\ref{lm:2}, $\bm{u}$ is a prefix-unit of $\bm{b}'$ and $\bm{u}'$ is a suffix-unit of $\bm{b}'$.

We first show that $\bm{u}=(P(\bm{\alpha}))^p$ for some $p\in\mathbb{Z}^{+}$.
Since $\bm{b}'=\bm{u}\circ\bm{b}$ and
$P(\bm{\alpha})\preceq_{\mathrm p}\bm{b}$,
we have
$\bm{u}\circ P(\bm{\alpha})\preceq_\mathrm{p}\bm{b}'$, and thus
$\alpha_{\ell(P(\bm{\alpha}))-1}=b'_{\ell(\bm{u}\circ P(\bm{\alpha}))-1}$.
If $\ell(\bm{u}\circ P(\bm{\alpha}))>\ell(\bm{\alpha})$, then
$\ell(\bm{b}')\geq\ell(\bm{u}\circ P(\bm{\alpha}))>\ell(\bm{\alpha})$,
which implies
$\bm{b}'=\bm{\alpha}\circ s^{\ell(\bm{b}')-\ell(\bm{\alpha})}$.
Consequently,
$\alpha_{\ell(P(\bm{\alpha}))-1}
=b'_{\ell(\bm{u}\circ P(\bm{\alpha}))-1}
=s$,
which contradicts
$\alpha_{\ell(P(\bm{\alpha}))-1}\neq s$
(Definition~\ref{cn}).
Therefore,
$\ell(\bm{u}\circ P(\bm{\alpha}))
\leq\ell(\bm{\alpha}),$
and thus
$\bm{u}\circ P(\bm{\alpha})\preceq_{\mathrm p}\bm{\alpha}$.
By Lemma~\ref{lemma:main}(2),
$\ell(\bm{u})\bmod\ell(P(\bm{\alpha}))=0$,
and thus
$\bm{u}=(P(\bm{\alpha}))^p$
for some $p\in\mathbb{Z}^{+}$.

Since
\(\bm b\prec_{\mathrm p}\bm b'\),
we have
\(\bm b'\neq P(\bm{\alpha})\).
We next consider the following two cases.
\begin{itemize}
   \item[(1)] $\bm{b}'\in\Phi$.

Since \(\bm b\prec_{\mathrm p}\bm b'\),
we have $\bm{b}\in\{P(\bm{\alpha})\}\cup\Phi$. 
Since $\bm{b}'\in\Phi$, we have $\ell(\bm{b}')\bmod\ell(P(\bm{\alpha}))\neq 0 $.
Since \(\bm{b}'=\bm{u}\circ\bm{b}\) and $\bm{u}=(P(\bm{\alpha}))^p$, we have $\ell(\bm{b})\bmod\ell(P(\bm{\alpha}))=\ell(\bm{b}')\bmod\ell(P(\bm{\alpha}))\neq 0$. Thus $\bm{b}\neq P(\bm{\alpha})$.

If $\bm{b}\in\Phi$, by Lemma~\ref{lm:01}(1), we have $\bm{\alpha}_{[0;\ell(\bm{\alpha})-\ell(P(\bm{\alpha}))+2]}\preceq_{\mathrm{p}}\bm{b}\prec_\mathrm{p}\bm{b}'\preceq_{\mathrm{p}}\bm{\alpha}$, and thus
$$\ell(\bm{b}')-\ell(\bm{b})
\leq\ell(\bm{\alpha})-(\ell(\bm{\alpha})-\ell(P(\bm{\alpha}))+2)=\ell(P(\bm{\alpha}))-2<\ell(P(\bm{\alpha})),$$
which contradicts $\ell(\bm{b}')-\ell(\bm{b})=\ell(\bm{u})\ge\ell(P(\bm{\alpha}))$.
Then, we have $\bm{b}\notin\Phi$. Therefore, for any $\bm{b}'\in\Phi$, there does not exist
$\bm{b}\in\mathcal{B}$ such that
$\bm{b}\prec_{\mathrm{s}}\bm{b}'.$

\item[(2)] $\bm{b}'\in\mathcal{S}$.

By the definition of $\mathcal{S}$, we have $\bm{b}'=\bm{\alpha}\circ s^{t}$ for some positive integer \(t\).
Then, the last symbol of \(\bm{b}'\) is \(s\), whereas the last symbol of \(P(\bm{\alpha})\) is \(\alpha_{\ell(P(\bm{\alpha}))-1}\neq s\).
Therefore,
$P(\bm{\alpha})$ cannot be a suffix of $\bm{b}'$.
Since $\bm{b}\prec_\mathrm{s}\bm{b}'$,
we have
$\bm{b}\neq P(\bm{\alpha})$,
and thus
$\bm{b}\in\Phi\cup\mathcal{S}$.
Then, $\bm{b}$ is either
$\bm{\alpha}_{[0;\ell(\bm{\alpha})-\ell(\bm{\varphi})]}$
for some
$\bm{\varphi}\in\Theta\setminus\{\bm{\theta}\}$,
or
$\bm{\alpha}\circ{s}^{t'}$
for some positive integer $t'$ satisfying $t'<t$.
Since $\bm{b}'=\bm{u}\circ\bm{b}=\bm{b}\circ\bm{u}'$, we have
\[
\bm{u}'=
\begin{cases}
\bm{\varphi}\circ s^{\ell(\bm{u})-\ell(\bm{\varphi})},
& \text{if } \bm{b}=\bm{\alpha}_{[0;\ell(\bm{\alpha})-\ell(\bm{\varphi})]},\\[1mm]
s^{\ell(\bm{u})},
& \text{if } \bm{b}=\bm{\alpha}\circ{s}^{t'}.
\end{cases}
\]

By Lemma~\ref{lm:0}(1), we have
$\bm{\varphi}\prec_{\mathrm{p}}P(\bm{\alpha})$, and thus $\bm{\varphi}\prec_{\mathrm p}\bm{u}$. 
On the other hand, since \(\bm{u}\) and \(\bm{u}'\) are units of \(\bm{b}'\) having the same length, by Lemma~\ref{lm:lemma2}(3), we have \(N_s(\bm{u})=N_s(\bm{u}')\).

If $\bm{u}'=\bm{\varphi}\circ s^{\ell(\bm{u})-\ell(\bm{\varphi})},$
then since \(\bm{\varphi}\prec_{\mathrm p}\bm{u}\) and $N_s(\bm{u})=N_s(\bm{u}')$, we have $\bm{u}=\bm{\varphi}\circ
s^{\ell(\bm{u})-\ell(\bm{\varphi})}$. Thus, \(\alpha_{\ell(P(\bm{\alpha}))-1}=u_{\ell(\bm{u})-1}=s\), which contradicts 
\(\alpha_{\ell(P(\bm{\alpha}))-1}\neq s\) (Definition~\ref{cn}).

If $\bm{u}'=s^{\ell(\bm{u})}$, then since $N_s(\bm{u})=N_s(\bm{u}')$, we have $\bm{u}=s^{\ell(\bm{u})}$.
Thus,
$\alpha_{\ell(P(\bm{\alpha}))-1}=u_{\ell(\bm{u})-1}=s,$
which contradicts
$\alpha_{\ell(P(\bm{\alpha}))-1}\neq s$
(Definition~\ref{cn}).

Therefore, for any $\bm{b}'\in\mathcal{S}$, there does not exist
$\bm{b}\in\mathcal{B}$ such that
$\bm{b}\prec_{\mathrm{s}}\bm{b}'.$
\end{itemize}

Together with Cases (1) and (2), we can conclude that $\mathcal{B}$ is suffix-free.
\end{IEEEproof}

\begin{corollary}\label{0o0}
Every sequence in $\mathcal{B}^{*}$ has a unique decomposition into a concatenation of elements of $\mathcal{B}$.
\end{corollary}

\begin{lemma}\label{lm:distinguish}
Letting $\bm{b} \in \mathcal{B}^{*}$ with $\ell(\bm{b}) \ge \ell(\bm{\alpha})$, we have  $\bm{\alpha}\preceq_\mathrm{p}\bm{b}$.
\end{lemma}

\begin{IEEEproof}
Let $\bm{u}$ be the first element in the decomposition of $\bm{b}$ over $\mathcal{B}$. 
We consider three cases.

\begin{itemize}
\item[(1)] $\bm{u}=P(\bm{\alpha})$.

Let ${p}$ denote the number of consecutive copies of
$P(\bm{\alpha})$ at the beginning of the decomposition of $\bm{b}$ over $\mathcal{B}$. 
Clearly, ${p}\ge1$. If ${p}\ell(P(\bm{\alpha}))\geq \ell(\bm{\alpha})$, then $\bm{\alpha}\preceq_\mathrm{p}(P(\bm{\alpha}))^p\preceq_\mathrm{p}\bm{b}$. Otherwise, letting $\bm{v}$ be the $({p}+1)$-th element in the decomposition of $\bm{b}$ over $\mathcal{B}$, we have $\bm{v}\in\Phi\cup\mathcal{S}$.
By Lemma~\ref{lm:01}(1), $\bm\alpha_{[0;\ell(\bm\alpha)-\ell(P(\bm\alpha))+2]}\preceq_\mathrm{p}\bm v$, and thus $\bm\alpha_{[0;\ell(\bm\alpha)-p\ell(P(\bm\alpha))]}\prec_\mathrm{p}\bm v$.
Then,
\[
\bm{b}_{[0;\ell(\bm\alpha)]}
=(P(\bm\alpha))^{{p}}\circ\bm v_{[0;\ell(\bm\alpha)-{{p}}\ell(P(\bm\alpha))]}
=(P(\bm\alpha))^{{p}}\circ\bm\alpha_{[0;\ell(\bm\alpha)-{{p}}\ell(P(\bm\alpha))]}
=(P(\bm\alpha))^{{p}}\circ\bm\alpha_{[{{p}}\ell(P(\bm\alpha));\ell(\bm\alpha)-{{p}}\ell(P(\bm\alpha))]}
=\bm\alpha.
\]

\item[(2)]
$\bm{u}\in\Phi$.

By the definition of $\Phi$, we have $\bm{u}=\bm{\alpha}_{[0;\ell(\bm{\alpha})-\ell(\bm{\varphi})]}$
for some $\bm{\varphi}\in\Theta\setminus\{\bm{\theta}\}$. 
By Lemma~\ref{lm:0}(1), we have $\bm{\varphi}\prec_{\mathrm{p}}P(\bm{\alpha})$. 
Let $\bm{v}$ be the second element in the decomposition
of $\bm{b}$ over $\mathcal{B}$.
By Definition~\ref{cn} and Lemma~\ref{lm:01}(2), we have $P(\bm{\alpha})\preceq_\mathrm{p}\bm{v},$ and thus $\bm{\varphi}\prec_\mathrm{p}\bm{v}.$
Then,
$$\bm{b}_{[0;\,\ell(\bm{\alpha})]}=\bm{u}\circ\bm{v}_{[0;\,\bm{\varphi})]}
=\bm{u}\circ\bm{\varphi}
=\bm{\alpha}.$$

\item[(3)]
$\bm{u}\in\mathcal{S}.$ 

Clearly, $\bm{\alpha}\prec_\mathrm{p}\bm{u}\preceq_\mathrm{p}\bm{b}$.
\end{itemize}
\end{IEEEproof}

With the above auxiliary results, we turn to the proof of
Theorem 1.

\begin{IEEEproof}[Proof of Theorem~1]
For an arbitrary but fixed integer
\(n\ge2\ell(\bm{\alpha})\), define
\[
\mathcal{C}_n
=
\{\bm{b}\circ\bm{\alpha}\mid
\bm{b}\in\mathcal{B}^*,\,
\ell(\bm{b}\circ\bm{\alpha})=n
\}.
\]

We first show that \(\mathcal{C}_n\subseteq\mathcal{B}^{*}\) and
\(\mathcal{C}_n\neq\emptyset\). To show that
\(\mathcal{C}_n\subseteq\mathcal{B}^{*}\), it suffices to prove that
\(\bm{\alpha}\in\mathcal{B}^{*}\). If \(\bm{\theta}=\bm{\varepsilon}\), then
\(\bm{\alpha}=(P(\bm{\alpha}))^w\in\mathcal{B}^*\).
Otherwise, \(\bm{\theta}\in\Theta\setminus\{\bm{\varepsilon}\}\),
and hence \(\bm{\alpha}\in\Phi\subseteq\mathcal{B}\).
Therefore \(\bm{\alpha}\in\mathcal{B}^{*}\).
Also,
$\bm{\alpha}\circ s^{n-2\ell(\bm{\alpha})}\circ\bm{\alpha}
\in\mathcal{C}_n,$
which shows that \(\mathcal{C}_n\neq\emptyset\).

We now show that \(\mathcal{C}_n\) is a code for the graph
\(G(\bm{\alpha})\), which provides a lower bound on
\(C_0(\bm{\alpha})\).
Let
\(\bm{c},\bm{c}'\in\mathcal{C}_{n}\)
be any two distinct sequences.
Consider their decompositions over \(\mathcal{B}\).
Look at the last position where the two decompositions differ. At that position, suppose $\bm{c}$ has $\bm{u}$ and $\bm{c}'$ has $\bm{u}'$, with $\bm{u},\bm{u}'\in\mathcal{B}$ and $\bm{u}\neq\bm{u}'$. Without loss of generality, assume that $\bm{u}\prec_\mathrm{p}\bm{u}'$. The strings after the position are identical in both sequences; denote this common suffix by $\bm{e}$. The prefixes before this position are denoted by $\bm{d}$ and $\bm{d}'$, respectively.
Then, 
\[
\bm{c}=\bm{d}\circ\bm{u}\circ\bm{e}
\quad\text{and}\quad
\bm{c}'=\bm{d}'\circ\bm{u}'\circ\bm{e}.
\]
Since \(\bm u\prec_{\mathrm p}\bm u'\), we can write
$\bm u'=\bm u''\circ\bm u'''$
with \(\ell(\bm u''')=\ell(\bm u)\). Thus,
$$\bm c'=\bm d'\circ\bm u''\circ\bm u'''\circ\bm e .$$
Clearly, $\bm{d},\bm{d}',\bm{u},\bm{u}',\bm{e}\in\mathcal{B}^*$, $\bm{\alpha}\preceq_{\mathrm{s}}\bm{e}$ and $\ell(\bm{d})=\ell(\bm{d}'\circ\bm{u}'')$.
By Lemma~\ref{CnB}, \(\bm{u}\) is not a suffix of
\(\bm{u}'\). Thus,
\begin{equation}\label{suffix}
\bm{u}\neq\bm{u}'''.
\end{equation}
We consider two cases to show that $\bm{c}$ and $\bm{c}'$ are distinguishable.

\begin{itemize}
\item[(1)] \(\ell(\bm{u})\leq\ell(\bm{\alpha})\).

We have
\[
\bm{\alpha}
\overset{(a)}{=}\bm{u}\circ
\bm{e}_{[0;\,\ell(\bm{\alpha})-\ell(\bm{u})]}
\overset{(b)}{\neq}
\bm{u}'''\circ\bm{e}_{[0;\,\ell(\bm{\alpha})-\ell(\bm{u})]},
\]
where \((a)\) follows from
\(\bm{u}\circ\bm{e}\in\mathcal{B}^{*}\) and Lemma~\ref{lm:distinguish},
and \((b)\) follows from \eqref{suffix}.
Then, by Lemma~\ref{co1},
\(\bm{c}\) and \(\bm{c}'\) are distinguishable for 
\(G(\bm{\alpha})\).

\item[(2)] \(\ell(\bm{u})>\ell(\bm{\alpha})\).

Clearly, $\bm{u},\bm{u}'\in\mathcal{S}$. Then,
$\bm{u}=\bm{\alpha}\circ s^{\ell(\bm{u})-\ell(\bm{\alpha})}$,
and
$\bm{u}'=\bm{u}''\circ\bm{u}'''=\bm{\alpha}\circ s^{\ell(\bm{u}')-\ell(\bm{\alpha})}.$
Thus,
$s=u_i=u_i''', \forall i\in\mathbb{Z}[\ell(\bm{\alpha}),\ell(\bm{u})-1].$
On the other hand, by \eqref{suffix}, there exists
\(j\in\mathbb{Z}[0,\ell(\bm{u})-1]\) such that
\(u_i\neq u_i'''\).
Therefore, $j\in\mathbb{Z}[0,\ell(\bm{\alpha})-1],$
and thus
$\bm{u}'''_{[0,\ell(\bm{\alpha})]}\neq\bm{u}_{[0,\ell(\bm{\alpha})]}=\bm{\alpha}$.
Then, by Lemma~\ref{co1},
\(\bm{c}\) and \(\bm{c}'\) are distinguishable for
\(G(\bm{\alpha})\).
\end{itemize}
Together with Cases (1) and (2), we can conclude that
\(\mathcal{C}_{n}\) is a code for \(G(\bm{\alpha})\).
    By Lemma~\ref{lem:rate_invariance} and a classical result of Shannon (cf. e.g.~\cite[Lemma 4.5]{book-it-janos}), we have
    $R(\{\mathcal{C}_{n}\})=-\log_2\mu,$
    where $\mu$ is the only positive root of the equation \[
x^{\ell({P}(\bm{\alpha}))}
+\sum_{\bm\phi \in \Phi}
x^{\ell(\bm{\phi})}
+\frac{x^{\ell(\bm{\alpha})+1}}{1-x}=1.
\]

We next show that
$C_0(\bm{\alpha})\leq -\log_2\mu.$
Recall the definition of \(\mathrm{L}_{\bm{\alpha}}(\cdot)\) in Definition~\ref{def:labeling_sequences}.
Let \(\{\mathcal{Y}_n\}\) be a sequence of sets indexed by \(n\), where $\mathcal{Y}_n=\{\mathrm{L}_{\bm{\alpha}}(\bm{x})\mid \bm{x}\in\mathcal{X}^{n}\}.$
Clearly, \(\mathcal{Y}_n\subseteq\{0,1\}^n\).
By Definition~\ref{rm1} and Lemma~\ref{co2},
$R(\{\mathcal{Y}_n\})=C_{\mathrm{L}}(\bm{\alpha})=C_0(\bm{\alpha}).$
Therefore, it suffices to show that
$R(\{\mathcal{Y}_{n}\})\leq -\log_2\mu.$

We prove this inequality by characterizing the constraints satisfied by the output sequences in \(\mathcal{Y}_n\).
Let \(\mathcal{Y}_{n}^{y_0y_1\cdots y_p}\) denote the subset of
\(\mathcal{Y}_n\) consisting of sequences with prefix
\(y_0y_1\cdots y_p\),
and let \(\bm{y}\in\mathcal{Y}_{n}^{1}\) be arbitrary but fixed. Then, there
exists \(\bm{x}\in\mathcal{X}^{n}\) such that
\(\mathrm{L}_{\bm{\alpha}}(\bm{x})=\bm{y}\).
Let $l$ be the coordinate of the second $1$ in $\bm{y}$, and thus \(\bm{y}\in\mathcal{Y}^{10^{l-1}1}_n\). By Definition~\ref{def:labeling_sequences}, we have
\begin{subequations}\label{eq:group01}
\begin{equation}
\bm{\alpha}=\bm{x}_{[0;\,\ell(\bm{\alpha})]}=\bm{x}_{[l;\,\ell(\bm{\alpha})]},
\label{eq:group01:a}
\end{equation}
\begin{equation}
\bm{x}_{[i;\,\ell(\bm{\alpha})]}\neq\bm{\alpha},
\forall i\in\mathbb{Z}[1,l-1].
\label{eq:group01:c}
\end{equation}
\end{subequations}

We now consider the case that $l\leq\ell(\bm{\alpha})$.
From \eqref{eq:group01:a}, we obtain $\bm{\alpha}_{[l;\ell(\bm{\alpha})-l]}
=\bm{x}_{[l;\ell(\bm{\alpha})-l]}$ 
and 
$\bm{\alpha}_{[0;\ell(\bm{\alpha})-l]}
=\bm{x}_{[l;\ell(\bm{\alpha})-l]},$
respectively, from
$\bm{\alpha}=\bm{x}_{[0,\ell(\bm{\alpha})]}$ and $\bm{\alpha}=\bm{x}_{[l,\ell(\bm{\alpha})]}$.
Thus, by Lemma~\ref{co22}, we have
$\bm{x}_{[l;\ell(\bm{\alpha})-l]}
\in\mathscr{S}''(\bm{\alpha})$. 
Note that $\bm{x}_{[0;l]}\circ\bm{x}_{[l;\ell(\bm{\alpha})-l]}=\bm{x}_{[0;\ell(\bm{\alpha})]}=\bm{\alpha}$. We further have $\bm{x}_{[0;l]}
\in\{(P(\bm{\alpha}))^k\mid k\in\mathbb{Z}[1,w]\}\cup\Phi.$
Moreover, if
$\bm{x}_{[0;l]}
\in\{(P(\bm{\alpha}))^k\mid k\in\mathbb{Z}[2,w]\},$
then since $\bm{x}_{[l;\ell(\bm{\alpha})]}=\bm{\alpha}$,
we have 
$\bm{x}_{[l-\ell(P(\bm{\alpha}));\,\ell(P(\bm{\alpha}))+\ell(\bm{\alpha})]}
=P(\bm{\alpha})\circ\bm{\alpha}\in\mathcal{B}^*.$ By Lemma~\ref{lm:distinguish}, we have $\bm{x}_{[l-\ell(P(\bm{\alpha}));\,\ell(\bm{\alpha})]}
=\bm{\alpha},$ 
which contradicts~\eqref{eq:group01:c}.
Therefore,
\begin{align*}
\bm{x}_{[0;l]}
\in(\{(P(\bm{\alpha}))^k\mid k\in\mathbb{Z}[1,w]\}\cup\Phi)
\setminus
\{(P(\bm{\alpha}))^k\mid k\in\mathbb{Z}[2,w]\}
=P(\bm{\alpha})\cup\Phi.
\end{align*}
Consequently, when $l\leq\ell(\bm{\alpha})$, we have
$l\in
\{\ell(P(\bm{\alpha}))\}
\cup
\{\ell(\bm{\phi})\mid\bm{\phi}\in\Phi\},$ and thus
$$l\in
\{\ell(P(\bm{\alpha}))\}
\cup
\{\ell(\bm{\phi})\mid\bm{\phi}\in\Phi\}
\cup
\{t\mid t>\ell(\bm{\alpha}),\,t\in\mathbb{Z}^{+}\}.$$
Then, we have
$$|{\mathcal{Y}}_n^1|
\leq\sum_l|{\mathcal{Y}}_n^{10^{l-1}1}|=
|{\mathcal{Y}}_{n-\ell({P}(\bm{\alpha}))}^1|
+\sum_{\bm{\phi} \in \Phi}
|{\mathcal{Y}}_{n-\ell(\bm{\phi})}^1|
+\sum_{t> \ell(\bm{\alpha})}|{\mathcal{Y}}_{n-t}^1|.$$
Hence,
$\lim\limits_{n \to \infty} \frac{1}{n}\log |{\mathcal{Y}}_n^1|\leq-\log\mu$,
where $\mu$ is the unique positive root of the equation:
\[
x^{\ell({P}(\bm{\alpha}))}
+\sum_{\bm{\phi} \in \Phi}x^{\ell(\bm{\phi})}
+\frac{x^{\ell(\bm{\alpha})+1}}{1-x}=1.
\]
Therefore,
\begin{align*}
    R(\{\mathcal{Y}_{n}\})
    &=\lim_{n \to \infty} \frac{\log |{\mathcal{Y}}_n|}{n}\\
    &=\lim_{n \to \infty} \frac{\log \left(|{\mathcal{Y}}_n^1|+|{\mathcal{Y}}_n^{01}|+|{\mathcal{Y}}_n^{001}|+\cdots\right)}{n}\\
    &\leq\lim_{n \to \infty} \frac{\log \left(n|{\mathcal{Y}}_n^1|\right)}{n}\\
    &=\lim_{n \to \infty} \frac{\log |{\mathcal{Y}}_n^1|}{n}\\
    &\leq-\log\mu.
\end{align*}
\end{IEEEproof}

Next, we characterize the label structures that achieve the minimum and maximum labeling capacities for an arbitrary but fixed label length.
\begin{theorem}\label{thm:main}
We have
$-\log\beta \leq C_0(\bm{\alpha}) \leq -\log\gamma,$
where $\beta$ and $\gamma$, respectively, are the unique positive roots of the equations
\[
x+x^{\ell(\bm{\alpha})}=1
\quad\text{and}\quad
x+x^\frac{\ell(\bm{\alpha})+1}{2}=1.
\]
Moreover, the lower bound is attained if and only if
$\ell(P(\bm{\alpha}))=\ell(\bm{\alpha}),$
and the upper bound is attained if and only if one of the following conditions holds:
\begin{enumerate}
    \item[(a)] $\ell(P(\bm{\alpha}))=1$;
    \item[(b)] $\ell(\bm{\alpha})=2\ell(P(\bm{\alpha}))-1$ and
    $\ell(P(\bm{\theta}))=1$.
\end{enumerate}
\end{theorem}
\begin{IEEEproof}
The inequality \(C_0(\bm{\alpha})\geq-\log\beta\) and its necessary and sufficient equality condition were established in Theorem~7 of~\cite{dnalable1}.
We now consider the upper bound.
By letting 
\begin{equation*}
f(x)=x+x^\frac{\ell(\bm{\alpha})+1}{2}
\end{equation*}
and
\begin{equation*}
g(x)=x^{\ell(P(\bm{\alpha}))}
+\sum_{\bm{\phi} \in \Phi}x^{\ell(\bm{\phi})}
+\frac{x^{\ell(\bm{\alpha})+1}}{1-x},
\end{equation*}
we have $f(\gamma)=1$ and 
$C_0(\bm{\alpha})=-\log\gamma_1$, where $\gamma_1$ is the unique positive root of $g(x)=1$.
Showing that $C_0(\bm{\alpha})\leq-\log\gamma$
is equivalent to showing that $\gamma_1\geq\gamma$.
Since $g(x)$ is strictly increasing for $x>0$,
this is further equivalent to showing that
$g(\gamma)\leq g(\gamma_1)=1$.

By Lemma~\ref{lm:01}(1) and Lemma~\ref{lm:01}(2), for any $\bm{\phi}\in\Phi$, we have $\ell(\bm{\phi})\in\mathbb{Z}[l,\ell(\bm{\alpha})]$, where
$l=\max\{\ell(P(\bm{\alpha}))+1,\ell(\bm{\alpha})-\ell(P(\bm{\alpha}))+2\}$. Thus,
\begin{equation}\label{you-}
    g(\gamma)
\overset{(a)}{\leq}\gamma^{\ell(P(\bm{\alpha}))}
+\sum_{i=l}^{\ell(\bm{\alpha})}\gamma^{i}
+\frac{\gamma^{\ell(\bm{\alpha})+1}}{1-\gamma}
=\gamma^{\ell(P(\bm{\alpha}))}
+\frac{\gamma^{l}}{1-\gamma}.
\end{equation}
where equality in $(a)$ holds only if 
\(|\Phi|=\ell(\bm{\alpha})-l+1\).
We now consider two cases.
\begin{itemize}
\item[(1)]
$l=\ell(P(\bm{\alpha}))+1>\ell(\bm{\alpha})-\ell(P(\bm{\alpha}))+2$, i.e.,
\(\ell(P(\bm{\alpha}))>\frac{\ell(\bm{\alpha})+1}{2}\).

From \eqref{you-} and \(0<\gamma<1\), we have
\begin{equation*}
g(\gamma)\leq\gamma^{\ell(P(\bm{\alpha}))}
+\frac{\gamma^{\ell(P(\bm{\alpha}))+1}}{1-\gamma}
=\frac{\gamma^{\ell(P(\bm{\alpha}))}}{1-\gamma}
<\frac{\gamma^{\frac{\ell(\bm{\alpha})+1}{2}}}{1-\gamma}
=\frac{f(\gamma)-\gamma}{1-\gamma}
=1.
\end{equation*}
Thus, \(g(\gamma)<1\) when \(\ell(P(\bm{\alpha}))>\frac{\ell(\bm{\alpha})+1}{2}\).

\item [(2)] 
$l=\ell(\bm{\alpha})-\ell(P(\bm{\alpha}))+2\ge\ell(P(\bm{\alpha}))+1$, i.e.,
\(\ell(P(\bm{\alpha}))\leq\frac{\ell(\bm{\alpha})+1}{2}\).

Letting \(h(x)=x+\frac{(1-\gamma)\gamma}{x}\), from \eqref{you-}, we have
\begin{align}\label{you-3}
    g(\gamma)
    &\leq \gamma^{\ell(P(\bm{\alpha}))}
    +\frac{\gamma^{\ell(\bm{\alpha})-\ell(P(\bm{\alpha}))+2}}{1-\gamma} \notag\\
    &=\gamma^{\ell(P(\bm{\alpha}))}
    +\frac{\gamma^{\ell(\bm{\alpha})+1}\gamma}
    {(1-\gamma)\gamma^{\ell(P(\bm{\alpha}))}} \notag\\
    &=\gamma^{\ell(P(\bm{\alpha}))}
    +\frac{(f(\gamma)-\gamma)^2\gamma}
    {(1-\gamma)\gamma^{\ell(P(\bm{\alpha}))}} \notag\\
    &=\gamma^{\ell(P(\bm{\alpha}))}
    +\frac{(1-\gamma)\gamma}
    {\gamma^{\ell(P(\bm{\alpha}))}} \notag\\
    &=h(\gamma^{\ell(P(\bm{\alpha}))}).
\end{align}
Since \(1\leq\ell(P(\bm{\alpha}))\leq\frac{\ell(\bm{\alpha})+1}{2}\) and \(0<\gamma<1\), we have
\(\gamma^{\ell(P(\bm{\alpha}))}\in[\gamma^{\frac{\ell(\bm{\alpha})+1}{2}},\gamma]\).
Moreover, since \((1-\gamma)\gamma>0\), \(h(x)\) is strictly convex on \(x>0\).
Therefore, by \eqref{you-3},
\(g(\gamma)\leq\max\{h(\gamma^{\frac{\ell(\bm{\alpha})+1}{2}}),h(\gamma)\}\).
A direct calculation gives
\(h(\gamma^{\frac{\ell(\bm{\alpha})+1}{2}})=h(\gamma)=1\).
Hence, $g(\gamma)\leq1$, where equality can hold only if
\begin{equation}\label{(1)}
    \ell(P(\bm{\alpha}))=\frac{\ell(\bm{\alpha})+1}{2}\quad\text{or}\quad\ell(P(\bm{\alpha}))=1.
\end{equation}
On the other hand, equality in \eqref{you-} can hold only if
$|\Phi|=\ell(\bm{\alpha})-l+1=\ell(P(\bm{\alpha}))-1$.
By Lemma~\ref{lm:01}(3), equality in~\eqref{you-} further requires
\begin{equation}\label{(2)}
    \ell(P(\bm{\alpha}))=1
\quad\text{or}\quad
\ell(P(\bm{\theta}))=1.
\end{equation}
Therefore, by \eqref{(1)} and \eqref{(2)}, $g(\gamma)=1$ only if one of the following conditions holds:
\begin{itemize}
    \item[(a)] $\ell(P(\bm{\alpha}))=1$;
    \item[(b)] $\ell(\bm{\alpha})=2\ell(P(\bm{\alpha}))-1$ and
    $\ell(P(\bm{\theta}))=1$.
\end{itemize}
By Theorem~1, these two conditions are also sufficient for
\(g(\gamma)=1\). Hence, \(g(\gamma)=1\) if and only if one of these two conditions holds.
\end{itemize}
\end{IEEEproof}

\section{Conclusion}\label{conclusion}
In this paper, we studied the labeling capacity in the single-label setting, which is equivalent to the zero-error capacity of star graphs. We completely characterized the zero-error capacity of all star graphs, thereby resolving the labeling capacity for all single-label cases. We also developed a general coding method for constructing capacity-achieving codes for all star graphs. Furthermore, for any fixed label length, we derived tight lower and upper bounds on the achievable labeling capacities and established necessary and sufficient conditions for a label structure to attain each bound. These results hold for arbitrary finite alphabets and are not restricted to the DNA alphabet. An interesting direction for future research is to extend the coding methods developed in this work to DNA labeling systems subject to practical constraints, such as run-length and \(\mathrm{GC}\)-content constraints.

\bibliographystyle{IEEEtran}       
\bibliography{reference}   

\end{document}